\documentclass[conference]{IEEEtran}
\IEEEoverridecommandlockouts

\usepackage{cite}
\usepackage{amsmath,amssymb,amsfonts}
\usepackage{algorithmic}
\usepackage{graphicx}
\usepackage{textcomp}
\usepackage{xcolor}
\def\BibTeX{{\rm B\kern-.05em{\sc i\kern-.025em b}\kern-.08em
		T\kern-.1667em\lower.7ex\hbox{E}\kern-.125emX}}
	
\usepackage[utf8]{inputenc}
\usepackage{textalpha} 
\usepackage{amsthm, latexsym}
\usepackage{ebproof} 
\usepackage{graphicx}
\usepackage[hidelinks]{hyperref}
\usepackage{xspace} 
\usepackage{xargs} 
\usepackage{mathdots} 
\usepackage{mdframed} 
\usepackage{enumitem} 

\theoremstyle{plain}
\newtheorem{theorem}{Theorem}
\newtheorem{lemma}[theorem]{Lemma}
\newtheorem{corollary}[theorem]{Corollary}

\newtheorem{proposition}[theorem]{Proposition}
\theoremstyle{definition}
\newtheorem{definition}[theorem]{Definition}

\newtheorem{remark}[theorem]{Remark}

\newcommand{\claim}[1]{\medskip\noindent\texttt{\underline{Claim #1}:}}
\newcommand{\claimproof}[1]{\smallskip\noindent\texttt{{Proof of Claim #1}:}}
\newcommand{\claimproofend}{\hfill$\dashv$\medskip}

\usepackage{comment}

\newcommand{\calA}{\mathcal{A}} 
\newcommand{\calB}{\mathcal{B}}

\newcommand{\calE}{\mathcal{E}}
\newcommand{\calK}{\mathcal{K}}
\newcommand{\calL}{\mathcal{L}}
\newcommand{\calM}{\mathcal{M}}
\newcommand{\calN}{\mathcal{N}}
\newcommand{\calT}{\mathcal{T}}

\newcommand{\calG}{\mathcal{G}}

\newcommand{\bbA}{\mathbb{A}} 
\newcommand{\bbB}{\mathbb{B}}

\newcommand{\bbS}{\mathbb{S}}

\newcommand{\bbG}{\mathbb{G}}

\newcommand{\sfN}{\mathsf{N}}

\newcommand{\sfR}{\mathsf{R}}
\newcommand{\sfS}{\mathsf{S}}
\newcommand{\sfT}{\mathsf{T}} 
\newcommand{\sff}{\mathsf{f}} 

\renewcommand{\|}{\mid}
\renewcommand{\phi}{\varphi}
\renewcommand{\epsilon}{\varepsilon}

\newcommand{\mybar}[1]{\overline{#1}}
\newcommand{\Nat}{\mathbb{N}}

\newcommand{\powset}{\mathcal{P}}

\newcommand{\ran}{\mathrm{ran}}

\newcommand{\isdef}{\mathrel{:=}}
\newcommand{\nada}{\varnothing}

\newcommand{\subst}[2]{[#1 / #2]} 
\newcommand{\impl}{\rightarrow} 
\newcommand{\liff}{\leftrightarrow} 
\newcommand{\Land}{\bigwedge}

\renewcommand{\land}{\wedge}
\renewcommand{\lor}{\vee}
\newcommand{\proves}{\vdash}

\newcommand{\Prop}{\mathsf{Prop}}
\newcommand{\Var}{\mathsf{Var}}
\newcommand{\FV}{\mathsf{FV}}
\newcommand{\Act}{\mathsf{Act}}
\newcommand{\Voc}{\mathsf{Voc}}

\newcommand{\sat}{\Vdash}
\newcommand{\modImpl}{\vDash}

\newcommand{\muML}{\calL_{\mu}}
\newcommand{\muMLtwo}{\calL_{\mu}^2}
\newcommand{\Clos}{\mathsf{Clos}}
\newcommand{\ClosN}{\mathsf{Clos}^{\neg}}
\newcommand{\Fix}{\mathsf{Fix}}
\newcommand{\nil}{\mathrm{nil}}

\newcommand{\lbox}{\scalebox{0.8}{$\square$}}

\newcommand{\lm}{\mbox{\ooalign{\ld \cr \hidewidth\raise.05ex\hbox{$* \mkern3.1mu$}\cr}}} 
\newcommand{\lbm}{\ooalign{$\lbox$ \cr \hidewidth\raise.05ex\hbox{$* \mkern5.5mu$}\cr}\hspace{-0.1cm}} \newcommand{\lcm}{\ooalign{$\lbox$ \cr \hidewidth\raise.05ex\hbox{$\cdot \mkern2.9mu$}\cr}}
\newcommand{\lboxp}[1][a]{\ensuremath{[#1]}\xspace}
\newcommand{\ldiap}[1][a]{\ensuremath{\langle#1\rangle}\xspace}

\newcommand{\Ru}{\ensuremath{\mathsf{R}}\xspace}
\newcommand{\AxLit}{\ensuremath{\mathsf{Ax1}}\xspace}

\newcommand{\AxBot}{\ensuremath{\mathsf{Ax2}}\xspace}
\newcommand{\RuOr}{\ensuremath{\mathsf{R}_{\lor}}\xspace}
\newcommand{\RuAnd}{\ensuremath{\mathsf{R}_{\land}}\xspace}

\newcommand{\RuDia}[1][\ensuremath{a}]{\ensuremath{\mathsf{\mathsf{R}_{\ldiap[#1]}}}\xspace}
\newcommand{\RuDiaL}[1][\ensuremath{a}]{\ensuremath{\mathsf{R}_{\ldiap[#1]}^l}\xspace}
\newcommand{\RuDiaR}[1][\ensuremath{a}]{\ensuremath{\mathsf{R}^r_{\ldiap[#1]}}\xspace}

\newcommand{\RuFp}[1]{\ensuremath{\mathsf{R}_{#1}}\xspace}
\newcommand{\RuMu}{\RuFp{\mu}}

\newcommand{\RuNu}{\RuFp{\nu}}

\newcommand{\RuEta}{\RuFp{\eta}}
\newcommand{\RuWeak}{\ensuremath{\mathsf{weak}}\xspace}
\newcommand{\RuExp}{\ensuremath{\mathsf{exp}}\xspace}

\newcommand{\RuDischarge}[1]{\ensuremath{\mathsf{D}_{#1}}\xspace}

\newcommand{\RuReset}[1][\dx]{\ensuremath{\mathsf{Reset_{#1}}}\xspace}
\newcommandx{\RuCompress}[2][1= ,2= ]{\ensuremath{\mathsf{Compress}_{#1}^{#2}}\xspace}

\newcommand{\RuThin}{\ensuremath{\mathsf{thin}}\xspace}

\newcommand{\ruleskip}{\vspace{0.4cm}} 
\newcommand{\gdist}{\hspace{0.08cm}|\hspace{0.08cm}} 

\newcommand{\Tokens}{\mathcal{D}}

\newcommand{\dx}{\ensuremath{\mathsf{x}}}
\newcommand{\dy}{\ensuremath{\mathsf{y}}}
\newcommand{\dz}{\ensuremath{\mathsf{z}}}
\newcommand{\dd}{\ensuremath{\mathsf{d}}}
\newcommand{\de}{\ensuremath{\mathsf{e}}}
\newcommand{\df}{\ensuremath{\mathsf{f}}}

\newcommand{\discharge}[2]{\ensuremath{\lceil #1 \rceil^{#2}}\xspace}

\newcommand{\downto}{\downharpoonright}
\newcommand{\fina}[1]{[#1]}

\newcommand{\NWtwo}{\ensuremath{\mathsf{NW^2}}\xspace}

\newcommand{\JStwo}{\ensuremath{\mathsf{JS^2}}\xspace}

\newcommand{\JStwoInfty}{\ensuremath{\mathsf{JS^2_{\infty}}}\xspace}

\newcommand{\scs}{\ensuremath{\mathrm{scs}}\xspace}
\newcommand{\cyclicPT}{\calT_{\pi}^C}

\newcommand{\green}{\mathrm{green}}
\newcommand{\white}{\mathrm{white}}

\newcommand{\conv}[1]{\breve{#1}}
\newcommand{\trace}[1][]{\rightsquigarrow_{#1}}
\newcommand{\ntrace}[1][]{\not\rightsquigarrow_{#1}}
\newcommand{\tracestep}{\to_{C}}
\renewcommand{\to}[1][]{\overset{#1}{\rightarrow}}

\newcommand{\AxTraceNeg}{\ensuremath{\mathsf{Ax3}}\xspace}
\newcommand{\AxTraceLoop}{\ensuremath{\mathsf{Ax4}}\xspace}
\newcommand{\RuTrans}{\ensuremath{\mathsf{trans}}\xspace}
\newcommand{\RuJump}{\ensuremath{\mathsf{jump}}\xspace}
\newcommand{\RuCut}{\ensuremath{\mathsf{cut}}\xspace}
\newcommand{\RuTCut}{\ensuremath{\mathsf{tcut}}\xspace}

\newcommand{\EG}{\mathcal{E}}
\newcommand{\Own}{O}
\newcommand{\WC}{W}
\newcommand{\Om}{\Omega}

\newcommand{\last}{\mathsf{last}}
\newcommand{\first}{\mathsf{first}}
\newcommand{\PM}{\mathit{PM}}

\newcommand{\eloi}{\exists}
\newcommand{\abel}{\forall}

\newcommand{\eClos}{\epsilon\mathrm{Clos}}

\newcommand{\Seq}{\mathsf{Seq}}
\newcommand{\Inst}{\mathsf{Inst}}
\newcommand{\conc}{\mathsf{conc}}

\title{Interpolation for the two-way modal $\mu$-calculus\thanks{The research of the first author has been made possible by a grant from the Dutch Research Council NWO, project nr. 617.001.857. The authors are grateful for comments by the anonymous referees.}}
\author{\IEEEauthorblockN{Johannes Kloibhofer}
	\IEEEauthorblockA{{University of Amsterdam}\\ Amsterdam, Netherlands}\\
	\and
	\IEEEauthorblockN{Yde Venema}
		\IEEEauthorblockA{{University of Amsterdam}\\ Amsterdam, Netherlands}\\
	}

\begin{document}
	
\maketitle
	
\begin{abstract}
The two-way modal $\mu$-calculus is the extension of the (standard) one-way 
$\mu$-calculus with converse (backward-looking) modalities.
For this logic we introduce two new sequent-style proof calculi: a non-wellfounded
system admitting infinite branches
and a finitary, cyclic version of this that employs annotations.

As is common in sequent systems for two-way modal logics, our calculi 
feature an analytic cut rule. 
What distinguishes our approach is the use of so-called trace atoms, 
which serve to apply Vardi's two-way automata in 
a proof-theoretic setting.

We prove soundness and completeness for both systems and subsequently use the 
cyclic calculus to show that the two-way $\mu$-calculus has the (local) Craig
interpolation property, with respect to both propositions and modalities.
Our proof uses a version of Maehara's method adapted to cyclic proof systems.
As a corollary we prove that the two-way $\mu$-calculus also enjoys
Beth's definability property.
\end{abstract}


\section{Introduction}\label{sec.introduction}

\paragraph{The modal $\mu$-calculus}

The (standard, one-way) modal $\mu$-calculus $\muML$ 
(\cite{brad:moda06,brad:muca18})
is an extension of propositional (poly-)modal logic with explicit least- and 
greatest fixpoint operators which allow the expression of various recursive
concepts in the language.
The version used today, which was introduced by Kozen~\cite{Kozen1983}, has 
become a key tool in the formal study of the behaviour of programs and the 
dynamics of processes in general.
Many temporal and dynamic logics such as the computational tree logics CTL and 
CTL${}^{*}$, and propositional dynamic logic PDL, have natural embeddings in 
$\muML$.
It has been shown that when it comes to bisimulation-invariant 
properties, $\muML$ has the same expressive power as monadic second-order 
logic~\cite{Janin96}.
Despite this increase in expressiveness, the $\mu$-calculus has remarkably good
computational properties, such as an \textsc{ExpTime}-complete satisfiablity 
problem~\cite{Emerson99}.

Furthermore, $\muML$ displays some excellent meta-logical behaviour.
For instance, it enjoys the finite model property (Kozen~\cite{koze:fini88}), 
and a strong, \emph{uniform interpolation} property as shown by D'Agostino and 
Hollenberg~\cite{dago:logi00}.
The set of valid $\muML$-formulas has an elegant Hilbert-style axiomatisation, 
already proposed by Kozen~\cite{Kozen1983} but only proved to be complete by 
Walukiewicz~\cite{Walukiewicz2000}.

The \emph{proof theory} of the modal $\mu$-calculus is generally considered to
be hard, largely due to the subtleties of nested least- and greatest fixpoint
operators.
Niwi\'{n}ski \& Walukiewicz~\cite{Niwinski1996} introduced, in a framework of
infinite games, a certain sequent system in which sequents can be derived by 
means of non-wellfounded proofs allowing infinite branches, as long as these 
satisfy a certain success (or `progress') condition.\footnote{%
   Niwi\'{n}ski \& Walukiewicz actually work with \emph{tableaux}, which are 
   somewhat dual to sequent systems by establishing
   satisfiability rather than validity of finite sets of formulas.
   However,
   tableau calculi and sequent-style proof systems have so much in common, that 
   in this paper we do not differentiate between the two approaches.
   } 
Using results from automata theory one may show that any such 
non-wellfounded proof can be replaced by one that is regular - it has only finitely many distinct subderivations.
Such a regular proof can be represented as the unfolding of a 
\emph{cyclic} proof, that is, a certain finite derivation tree with back edges.
As a further step, Jungteerapanitch~\cite{Jungteerapanich2010} and 
Stirling~\cite{Stirling2014} annotate formulas to encode the automata-theoretic information,
which guarantees that all infinite branches in a proof 
tree satisfy the mentioned success condition.
We shall build on this work.

\paragraph{The two-way $\mu$-calculus}
Despite its expressive power, \emph{enrichments} of the standard modal
$\mu$-calculus have been introduced with various features that are well-known 
in the setting of basic (that is, fixpoint-free) modal logic, such as converse
modalities, nominals or counting modalities~\cite{bona:comp08}.
In this paper we shall focus on the \emph{two-way $\mu$-calculus} 
$\muMLtwo$~\cite{Vardi1998}, perhaps the most natural and simple extension of 
$\muML$.
Concretely, we obtain the language of $\muMLtwo$ by adding, for each modality $a$, a modality $\conv{a}$ which in the semantics will be interpreted as 
the converse of the accessibility relation for $a$.
This addition enables $\muMLtwo$ to reason about the past, which is attractive 
from the perspective of formal program verification~\cite{lich:glor85}, but also
in the area of description logics, where converse modalities correspond to 
inverse roles~\cite{degi:desc94}.

Compared to its one-way version, surprisingly little seems to be known about 
this logic.
While it is not hard to see that the two-way $\mu$-calculus lacks the finite
model property, a key result by Vardi~\cite{Vardi1998} states that the 
satisfiablity problem for the two-way $\mu$-calculus can be solved in 
exponential time. 
Vardi introduces the notion of an alternating two-way tree automaton, and the 
key argument in his proof is based on a reduction of these two-way automata 
to one-way deterministic tree automata.
Building on this, French~\cite{fren:bisi06,fren:idem07} proves various results 
on bisimulation quantifiers for $\muMLtwo$ via a transfer of results for $\muML$.
Afshari, J\"ager \& Leigh proposed a sound, complete and cut-free derivation 
system, which features an infinitely branching proof rule~\cite{afsh:infi19}.
A finitary, cyclic proof calculus was given by Afshari et al.~\cite{afsh:proo23};
this calculus is not cut-free, but its restrictions on the cut rule make the
system suitable for proof search procedures.
\smallskip

\textbf{Our goal} is to contribute to the knowledge on the two-way 
$\mu$-calculus by showing that it has two important properties: Craig 
interpolation and Beth definability.
Our main tool for establishing this, is a new, cyclic proof system for 
$\muMLtwo$ that is of independent interest.

\paragraph{Craig interpolation and Beth definability}
A logic has the Craig Interpolation Property (CIP) if for any pair of formulas 
$\phi$ and $\psi$ such that $\psi$ is a consequence of $\phi$ (denoted as 
$\phi \modImpl \psi$), one may find an \emph{interpolant}; that is, a formula
$\theta$ in the common vocabulary of $\phi$ and $\psi$ such that $\phi \modImpl
\theta$ and $\theta \modImpl \psi$.
The property was named of William Craig who proved it for first-order 
logic~\cite{crai:thre57}, since then it has been studied intensively in logic
(see for instance van Benthem~\cite{bent:many08}) while it also has many computer
science applications in for instance knowledge representation~\cite{jung:sepa21}
and model checking~\cite{mcmi:appl05}.

Beth Definability is a related property; informally it states that the implicit
definability of a concept in a logic implies the existence of an explicit
definition.
This notion has its origin in the work of Beth~\cite{beth:pao53} who discovered 
it as a property of classical first-order logic.
Beth definability is studied intensively in the area of description logics,
where it is used to optimize reasoning ~\cite{cate:beth13}. 

The main result in our paper is the following.

\begin{theorem}
The two-way modal $\mu$-calculus has the Craig interpolation property and, 
as a corollary, the Beth definability property.
\end{theorem}

We now sketch our proof for this result, 
which is proof-theoretic in nature.

\paragraph{Our proof systems}

The basis of our proof system is the annotated cyclic system of 
Jungteerapanitch~\cite{Jungteerapanich2010} and Stirling~\cite{Stirling2014} for $\muML$,
which itself builds on the non-wellfounded tableau system of Niwi\'{n}ski \& 
Walukiewicz~\cite{Niwinski1996}.
The extension with backward modalities causes two kinds of challenges.

The first complication is that, already in the setting without fixpoint 
operators, cut-free derivation systems for modal logics with backward
modalities have to go beyond simple sequent systems~\cite{ohni:gent57}.
One solution to this problem is to take resort to more structured sequents,
such as the nested sequents of Kashima~\cite{kash:cutf94}.
This road, however, also taken by Afshari et al. \cite{afsh:infi19}, does not combine well with cyclic proofs, since there is no bound on the
number of possible sequents.
Alternatively, one may simply allow the cut rule: applications of cut that 
are \emph{analytic} (in the wide sense that the cut formula must be taken 
from some bounded set of formulas) fit well in cyclic proofs.
Afshari et al.~\cite{afsh:proo23} took this approach, and so will we.

The second and main challenge is to formulate adequate success conditions on
infinite proof branches.
The problem is that the combinatorics of the formula traces are more complicated 
than in the one-way setting, since traces may move both up and down a proof 
tree. 
In order to deal with this issue, we follow Rooduijn \& Venema~\cite{rood:focu23}, 
who enrich the syntax of their proof calculus with so-called \emph{trace atoms}.
Roughly speaking, trace atoms hardwire the ideas underlying
Vardi's two-way automata explicitly into the syntax.
Rooduijn \& Venema restricted attention to the \emph{alternation-free fragment}
of $\muMLtwo$, in which entanglement of least- and greatest fixpoint operators
is basically avoided.
Our contribution here is to extend their approach to the technically
more challenging setting of the full two-way $\mu$-calculus.

\paragraph{Interpolation for the two-way $\mu$-calculus}
As mentioned, our approach is \emph{proof-theoretic}; more specifically, we 
employ a version of \emph{Maehara's method}~\cite{maeh:crai61}.
In this approach, in order to obtain an interpolant for an implication, one 
takes some finite proof for it and defines interpolants for each
node of the derivation tree, by means of a leaf-to-root induction.
It has been shown by Kowalski \& Ono~\cite{kowa:anal17} that the presence of 
an analytic cut rule does not preclude the application of Maehara's methods.

In recent years, the scope of the method has been extended to include cyclic
proofs.
Shamkanov~\cite{Shamkanov2014}, Afshari \& Leigh~\cite{Afshari2019} and Marti 
\& Venema~\cite{mart:focu21} used it to prove interpolation properties for,
respectively, G\"odel-L\"ob logic, $\muML$ and its alternation-free fragment.
Roughly, the challenge here is that in a cyclic proof some proof leaves are
not axiomatic and hence fail to have a trivial interpolant.
However, each such leaf is discharged at a companion node, closer to the root.
The idea is now to associate, as a kind of pre-interpolant, a fixpoint variable
with each discharged leaf, and to bind this variable at the companion with a
fixpoint operator.


The main technical novelty in our approach is to make this method work for our
proof system, which is not only cyclic, but also features a cut rule, and 
operates a nonstandard syntax including trace atoms.

\paragraph{Related work}
Other proof systems for $\muMLtwo$ and its alternation-free fragment were introduced in \cite{afsh:infi19}, \cite{afsh:proo23} and \cite{rood:focu23}.
There is also work by Benedikt and collaborators on \emph{guarded fixpoint 
logics}~\cite{bene:inte15,bene:defi19}, formalisms that \emph{extend} $\muMLtwo$ 
in expressive power.
Their approach is model-theoretic in nature.


\section{Preliminaries}\label{sec.prelim}

\subsection{Two-way modal $\mu$-calculus}

\paragraph{Syntax}
Let $\Prop$ be a set of propositions, $\Var$ a set of variables and $\Act$ a set
of actions. 
We assume an involution operation $\conv{\cdot}: \Act \to \Act$ such that for every $a \in \Act$ it holds that $a \neq \conv{a}$ and $a = \conv{\conv{a}}$.
The set $\muMLtwo$ of \emph{formulas} of the two-way modal $\mu$-calculus is generated 
by the grammar
\[
\phi \isdef
\bot \gdist \top \gdist
p \gdist \mybar{p} \gdist x
\gdist \phi\lor\phi \gdist \phi\land\phi \gdist
\ldiap\phi \gdist \lboxp\phi 
\gdist \mu x . \phi \gdist \nu x . \phi,
\]
where $p \in \Prop$, $x \in \Var$ and $a \in \Act$.

We define the \emph{vocabulary of $\phi$}, written $\Voc(\phi)$, to be the set 
of propositions and actions occurring in $\phi$ (with the proviso that we
include both $a$ and $\conv{a}$ in the vocabulary of any expression in which
$a$ or $\conv{a}$ occurs).

We refer to formulas of the form $\mu x . \phi$ and $\nu x . \phi$ as $\mu$- and
\emph{$\nu$-formulas}, respectively; formulas of either kind are called 
\emph{fixpoint formulas}, and the symbols $\mu$ and $\nu$ themselves are 
\emph{fixpoint operators}. 
This terminology is in line with the intended semantics, where $\mu x. \phi$
and $\nu x. \phi$ describe, respectively, the least and greatest fixpoint of
$\phi$ with respect to $x$.
We use $\eta,\lambda \in \{ \mu, \nu \}$ to denote arbitrary fixpoint 
operators.

Formulas that do not contain fixpoint operators are called \emph{fixpoint-free formulas}.
We use standard terminology and notation for the binding of variables by the
fixpoint operators and for the substitution operation.
We define $\FV(\phi)$ to be the set of free variables occurring in a formula $\phi$.
We call a formula $\phi$ a \emph{sentence} if all variable occurrences in $\phi$
are bound. 
Unless otherwise noted we will assume that every formula is a sentence.
We make sure only to apply substitution in situations where no variable capture
will occur. 
An important use of the substitution operation concerns the \emph{unfolding}
$\chi[\xi/x]$ of a fixpoint formula $\xi = \eta x . \chi$.

Given two formulas $\phi,\psi \in \muMLtwo$ we write $\phi \tracestep \psi$ if 
either $\phi$ is boolean or modal and $\psi$ is a direct subformula of
$\phi$, or else $\phi$ is a fixpoint formula and $\psi$ is its unfolding.
The \emph{closure} $\Clos(\Phi) \subseteq \muMLtwo$ of $\Phi \subseteq \muMLtwo$
is the least superset of $\Phi$ that is closed under this relation.
It is well known that $\Clos(\Phi)$ is finite iff $\Phi$ is finite.
A \emph{trace} is a sequence $(\phi_{n})_{n<\kappa}$, with $\kappa \leq
\omega$, such that $\phi_{n} \tracestep \phi_{n+1}$, for all $n
+ 1 < \kappa$.
It is well known that any infinite trace $\tau = (\phi_{n})_{n<\omega}$ features a unique formula 
$\phi$ that occurs infinitely often on $\tau$ and is
a subformula of $\phi_{n}$ for cofinitely many $n$. This formula is always a 
fixpoint formula, and where it is of the form $\eta x.\psi$ we call $\tau$ an
\emph{$\eta$-trace}.

We define a \emph{dependence order} on the fixpoint formulas in $\Phi$,
written $\Fix(\Phi)$, by setting $\eta x. \phi 
\geq_{\Phi} \lambda y. \psi$ (where bigger in $\geq_{\Phi}$ means being of higher 
priority) if $\Clos(\eta x.\phi) = \Clos(\lambda y.\psi)$ and $\eta x.\phi$ is a subformula of $\lambda y.\psi$.
Let $\Nat^+ \isdef \Nat\setminus\{0\}$. One may define a \emph{priority function} $\Omega: \Fix(\Phi) \rightarrow
\Nat^+$, which respects this order (namely, $\Omega(\eta x. \phi)\geq\Omega(\lambda y.
\psi)$ if $\eta x. \phi \geq_{\Phi} \lambda y \psi$) and satisfies that
$\Omega(\eta x. \phi)$ is even iff $\eta = \mu$. 
We extend $\Omega$ to a function $\Omega: \Phi \to \Nat^+$ by setting $\Omega(\phi) = 1$ if $\phi$ is not a fixpoint formula.

Given a formula $\phi \in \muMLtwo$ we inductively extend the map $p \mapsto 
\mybar{p}$ to a full-blown negation operation:
\[\begin{array}{lllclllclll}
\mybar{\bot} & \isdef & \top 
  && \mybar{\phi \land \psi} & \isdef & \mybar{\phi} \lor \mybar{\psi} 
  && \mybar{\mu x. \phi} & \isdef & \nu x. \mybar{\phi}
\\ \mybar{\top} & \isdef & \bot 
  && \mybar{\phi \lor \psi} & \isdef & \mybar{\phi} \land \mybar{\psi}
  && \mybar{\nu x. \phi} & \isdef & \mu x. \mybar{\phi}
\\ \mybar{x} & \isdef & x 
  && \mybar{\lboxp \phi} & \isdef & \ldiap \mybar{\phi}
\\ \mybar{\mybar{p}} & \isdef & p  
  && \mybar{\ldiap \phi} & \isdef & \lboxp \mybar{\phi}  
\end{array}\]

Note that $\mybar{\mybar{\phi}} = \phi$ for every formula $\phi$. For a set of formulas $\Phi$ we define $\mybar{\Phi}= \{\mybar{\phi} \| \phi \in \Phi\}$. We define $\ClosN(\Phi) \isdef \Clos(\Phi) \cup \Clos(\mybar{\Phi})$. 

\smallskip

For this work we fix a finite set of $\muMLtwo$ formulas $\Phi$, that is closed under $\tracestep$ and negation, in other words such that $\ClosN(\Phi) = \Phi$. For a proof of  a sequent $\Gamma$ the set $\Phi$ can be defined as $\Clos(\Gamma) \cup \Clos(\mybar{\Gamma})$. We also fix the priority function $\Omega: \Phi \to \Nat$ and let $n$ be the maximal even number in\footnote{For any function $f: A \to B$ we define $\ran(f) \isdef \{f(x) \| x \in A\}$.} $\ran(\Omega)$.

\smallskip

\paragraph{Semantics}

We assume some basic familiarity with \emph{parity games}. Definitions and notations that we use can be found in Appendix \ref{app.parityGames}.

A \emph{Kripke model} is a tuple $\bbS = (S,R,V)$, where $S$ is a non-empty set, $R= \{R_a \subseteq S\times S \| a \in \Act\}$ is a family of binary relation on $S$, such that $(s,s') \in R_a$ iff $(s',s) \in R_{\conv{a}}$, and $V$ is a function $\Prop \rightarrow \powset(S)$. A pair of a Kripke model $\bbS$ and a state $s \in \bbS$ is called a \emph{pointed model}.

The meaning of $\muMLtwo$-formulas on Kripke models can be given, as for the one-way $\mu$-calculus, by algebraic semantics or game semantics. We define the latter and refer for a definition of the former and a proof of its equivalence to \cite{Demri2016}.

Let $\xi \in \muMLtwo$ and let $\bbS = (S,R,V)$ be a Kripke model. The \emph{evaluation game} $\EG(\xi,\bbS)$ is the following infinite two-player 
game.
Its positions are pairs of the form $(\phi,s) \in \Clos(\xi)\times S$, and its
ownership function and admissible moves are given in Table~\ref{tb:EG}.
Infinite matches of the form
$(\phi_{n},s_{n})_{n<\omega}$ are won by $\eloi$ if the induced trace $(\phi_{n})_{n<\omega}$ is a $\nu$-trace, and won by $\abel$ if it is a $\mu$-trace.
It is well-known that this game can be presented as a parity game, and as such
it is positionally determined.
\begin{table}[htb]
\begin{center}
		\begin{tabular}{|ll|c|l|}
			\hline
			\multicolumn{2}{|l|}{Position} & Owner & Admissible moves\\
			\hline
			$(p,s)$        & with  $s \in V(p)$         
			& $\abel$ & $\nada$ 
			\\   $(p,s)$        & with  $s \notin V(p)$      
			& $\eloi$ & $\nada$ 
			\\   $(\mybar{p},s)$   & with $s \in V(p)$    
			& $\eloi$ & $\nada$ 
			\\   $(\mybar{p},s)$  & with $s \notin V(p)$ 
			& $\abel$ & $\nada$ 
			\\ \multicolumn{2}{|l|}{$(\phi \lor \psi,s)$}   & $\eloi$   
			& $\{ (\phi,s), (\psi,s) \}$ 
			\\ \multicolumn{2}{|l|}{$(\phi \land \psi,s)$} & $\abel$ 
			& $\{ (\phi,s), (\psi,s) \}$ 
			\\ \multicolumn{2}{|l|}{$(\ldiap \phi,s)$}        & $\eloi$ 
			& $\{ (\phi,t) \mid (s,t) \in R_a \}$ 
			\\ \multicolumn{2}{|l|}{$(\lboxp \phi,s) $}       & $\abel$ 
			& $\{ (\phi,t) \mid (s,t) \in R_a\}$ 
			\\ \multicolumn{2}{|l|}{$(\eta x . \phi,s)$}    & - 
			& $\{ (\phi[\eta x. \phi/x],s) \}$ 
			\\ \hline
		\end{tabular}
\end{center}
\caption{The evaluation game $\EG(\xi,\bbS)$}
\label{tb:EG}
\end{table}

\begin{definition}
Let $\bbS, s$ be a pointed model, let $f$ be a strategy for $\eloi$ in 
$\calE(\Land \Phi, \bbS)$ and let $\phi$ be a formula.
We write $\bbS, s \sat_f \phi$ if $f$ is winning for $\eloi$ at $(\phi,s)$.
\end{definition}

\begin{definition}
For any pair of formulas $\phi, \psi$ and $k = 1,\ldots,n$ we let $\phi \trace[k] \psi$ be the \emph{trace atom of priority $k$} and $\phi \ntrace[k] \psi$ be the \emph{negated trace atom of priority $k$}.
\end{definition}

\begin{definition}
	Given a strategy $f$ for $\eloi$ in $\calE(\Land \Phi, \bbS)$, we say that $\phi 
	\trace[k] \psi$ is \emph{satisfied} in $\bbS$ at $s$ with respect 
	to $f$, written $\bbS, s \sat_f \phi \trace[k] \psi$ if there is an $f$-guided match
	\[
	(\phi,s) = (\phi_0,s_0) \cdots (\phi_n,s_n) = (\psi,s), \qquad n > 0
	\] 
	such that $k = \max\{\Omega(\phi_i) \| i =0,\ldots,n-1\}$.
\end{definition}

A \emph{pure sequent} is a finite set of formulas, a \emph{trace sequent} is a
finite set of trace atoms, and a \emph{sequent} is a pure sequent together with
a trace sequent. 
We will use letters $A,B,\ldots$ as variables ranging over formulas and trace
atoms; $\Gamma, \Delta, \Sigma,\ldots$ for sequents and state explicitly if a 
sequent is a pure or a trace sequent. 
We define $\lboxp \Gamma = \{\lboxp \phi \| \phi \in \Gamma\}$.
Given a sequent $\Gamma$ we define $\Clos(\Gamma) \isdef \Clos(\{\phi \in \muMLtwo \| \phi \in \Gamma\})$ and analogously for $\ClosN$.

Note that our perspective on sequents and derivations is partly tableau-theoretic: We read sequents \emph{conjunctively} and aim to derive sequents that are \emph{unsatisfiable}. 
\begin{definition}
We say that $\bbS,s \sat_f \Gamma$ if $\bbS,s 
\sat_f A$ for all $A \in \Gamma$.
We define $\bbS,s \sat \Gamma$ if $\bbS, s \sat_f \Gamma$ for some strategy $f$ for $\eloi$ in $\calE(\Land \Phi, \bbS)$. 
A sequent $\Gamma$ is \emph{satisfiable} if $\bbS,s \sat \Gamma$ for some pointed model $\bbS,s$ and \emph{unsatisfiable} otherwise.
\end{definition}

We write  $\phi \modImpl \psi$ for the \emph{local consequence relation} meaning that $\bbS,s \sat \phi$ implies $\bbS,s \sat \psi$ for every pointed model $\bbS,s$.

\subsection{Stream automata with $\epsilon$-relations}

We define automata operating on streams (infinite words).
In addition to basic transitions we allow $\epsilon$-transitions: transitions without
an input letter. We call those automata \emph{$\epsilon$-automata}.

	Let $\Sigma$ be a finite set, called an \emph{alphabet}.
	An \emph{$\epsilon$-automaton} over $\Sigma$ is a quadruple 
	$\bbA = \langle A, \Delta, a_I,\mathrm{Acc}\rangle$, where $A$ is a 
	finite set; $\Delta = \Delta_b \cup \Delta_\epsilon$ is the transition function of $\bbA$ split up in a set of \emph{basic transitions} $\Delta_b: A \times \Sigma \rightarrow  \powset(A)$ and a set of \emph{$\epsilon$-transitions} $\Delta_\epsilon: A \rightarrow \powset(A)$; $a_I \in A$ its initial state; and $\mathrm{Acc} 
	\subseteq A^{\omega}$ its acceptance condition. 	
	An $\epsilon$-automaton is called \emph{deterministic} if $|\Delta_b(a,y)| = 1$ for all pairs 
	$(a,y) \in A \times \Sigma$ and $\Delta_\epsilon(a)$ is empty for all $a \in A$. 
	We assume that $\Delta_\epsilon(a_I) = \nada$.
	
	A \emph{run} of such an $\epsilon$-automaton $\bbA$ on a stream $w=z_0z_1z_2\cdots \in 
	\Sigma^{\omega}$ is a stream $(a_0,n_0)(a_1,n_1)(a_2,n_2)\cdots \in 
	(A\times \Nat)^{\omega}$ such that $(a_0,n_0) = (a_I,0)$, and for all 
	$i \in \omega$ either (i) $a_{i+1} \in \Delta_b(a_i,z_{n_i})$ and $n_{i+1} = n_i + 1$ or (ii) $a_{i+1}  \in \Delta_\epsilon(a_i)$ and $n_{i+1} = n_i$, and $\sup\{n_i \| i \in \omega\} = \omega$. The last condition guarantees that every run contains infinitely many basic transitions; in other words, we do not allow runs that from some point onwards only consist of $\epsilon$-transitions.
	
	The \emph{projection} of a run $(a_0,n_0)(a_1,n_1)(a_2,n_2)\cdots$ is the
	stream $a_0a_1a_2\cdots \in A^\omega$.
	A stream $w$ is \emph{accepted} by $\bbA$ if there is a run of $\bbA$ on $w$, 
	whose projection is in $\mathrm{Acc}$.
	We define $\calL(\bbA)$ to be the set of all accepting streams of $\bbA$.

	The acceptance condition can be given in different ways:
		A \emph{Büchi} condition is given as a subset $F \subseteq A$. 
		The corresponding acceptance condition is the set of runs, which contain 
		infinitely many states in $F$.
		A \emph{parity} condition is given as a map $\Omega: A \rightarrow \omega$. The corresponding acceptance condition is the set of runs $\alpha$ such that $\max\{\Omega(a) \| a \text{ occurs infinitely often in } \alpha \}$ is even.
		A \emph{Rabin} condition is given as a set $R = ((G_i,B_i))_{i \in I}$ of pairs of subsets of $A$.  The corresponding acceptance condition is the set of runs $\alpha$ for which there exists $i \in I$ such that $\alpha$ contains infinitely many states in $G_i$ and finitely many in $B_i$.
	$\epsilon$-Automata with these acceptance conditions are called \emph{Büchi}, \emph{parity} 
	and \emph{Rabin automata}, respectively.

\section{Trail-based proof system}

We introduce the non-wellfounded proof system \NWtwo, which generalizes tableau games for the one-way modal $\mu$-calculus 
introduced by Niwi\'{n}ski \& Walukiewicz \cite{Niwinski1996}. 
They construct a tree-shaped tableau $T$ for every unsatisfiable sequent
$\Gamma$, where every infinite path in $T$ carries a $\mu$-trail: an 
ancestry-path of formulas, where the most important formula occurring 
infinitely often is a $\mu$-formula. 
For the two-way case we will introduce a similar notion, but now trails may 
go up and down in $T$. As $T$ has the shape of a tree we can split up 
an infinite trail into segments that return to the same node and those only 
going up. We will call the former \emph{detour trails} and model them 
with trace atoms. Detour trails at a node $v$ talk about possible trails further up in the tree, hence we need a \RuCut rule in the proof system to model that behaviour. This approach stems from Vardi's two-way automata \cite{Vardi1998} and inspired Rooduijn \& Venema \cite{rood:focu23} to develop a proof system for the alternation-free fragment of the modal $\mu$-calculus.

\subsection{\NWtwo proofs}

Rules have the form $\begin{prooftree}
	\hypo{\Gamma_1}
	\hypo{\cdots}
	\hypo{\Gamma_n}
	\infer[left label= \Ru:]3{\Gamma}
\end{prooftree}$, where $\Gamma, \Gamma_1,\ldots,\Gamma_n$ are sequents. 
We call $\Gamma_1,\cdots,\Gamma_n$ the \emph{premisses} of \Ru and $\Gamma$ its \emph{conclusion}.

The rules of the \NWtwo proof system are given in Figure \ref{fig.NWtwo}.
Apart from the inclusion of trace atoms these rules coincide with the tableau
rules introduced in \cite{Niwinski1996} for the one-way $\mu$-calculus, except
that formulas are added in the premiss of \RuDia. 
Trace atoms and the extra rules \RuCut, \RuTCut and \RuTrans are added to deal
with converse modalities.

In the rules \RuAnd, \RuOr, \RuEta and \RuDia the single explicitly written
formula in the conclusion is called the \emph{principal formula}. 
The other rules do not have a principal formula.

\begin{definition}\label{def.ModalRule}
In order to define the modal rule \RuDia, let $\Psi = \ldiap \phi, \lboxp \Sigma,
\Gamma$ be the conclusion of the modal rule. 
We demand that $\Sigma$ is a pure sequent, and define $\ldiap[\conv{a}] \Gamma \isdef \{\ldiap[\conv{a}] \gamma \in \ClosN(\Psi)
\| \gamma \in \Gamma  \}$  and
	\begin{align*}
		\Gamma^{\ldiap \phi} 
		\isdef& ~\{\phi \ntrace[k]  \lboxp[\conv{a}] \chi \| \ldiap \phi \ntrace[k]  \chi \in \Gamma \text{ and } \lboxp[\conv{a}]\chi \in \ClosN(\Psi)\} \\
		\cup& ~\{\lboxp[\conv{a}] \chi \trace[k]  \phi \| \chi \trace[k]  \ldiap \phi \in \Gamma \text{ and } \lboxp[\conv{a}]\chi \in \ClosN(\Psi)\} \\
		\cup& ~\{\psi \ntrace[k]  \lboxp[\conv{a}] \chi \| \lboxp \psi \ntrace[k]  \chi \in \Gamma \text{ and } \lboxp[\conv{a}]\chi \in \ClosN(\Psi)\} \\
		\cup& ~\{\lboxp[\conv{a}] \chi \trace[k]  \psi \| \chi \trace[k]  \lboxp \psi \in \Gamma \text{ and } \lboxp[\conv{a}]\chi \in \ClosN(\Psi)\} 
	\end{align*}
	Note that $\ldiap[\conv{a}] \Gamma$ and $\Gamma^{\ldiap \phi}$ also depend on $\Psi$, yet for simpler notation we omit the extra subscript.
\end{definition}

\begin{figure*}[tbh]
	\begin{mdframed}[align=center]
		\begin{minipage}{\textwidth}
			\begin{minipage}{0.21\textwidth}
				\begin{prooftree}
					\hypo{\phantom{X}}
					\infer[left label =\AxLit:]1{\phi, \mybar{\phi}, \Gamma}
				\end{prooftree}
			\end{minipage}
			\begin{minipage}{0.18\textwidth}
				\begin{prooftree}
					\hypo{\phantom{X}}
					\infer[left label =\AxBot:]1{\bot, \Gamma}
				\end{prooftree}
			\end{minipage}
			\begin{minipage}{0.35\textwidth}
				\begin{prooftree}
					\hypo{\phantom{X}}
					\infer[left label =\AxTraceNeg:]1{\phi \trace[k] \psi, \phi \ntrace[k] \psi, \Gamma}
				\end{prooftree}
			\end{minipage}
			\begin{minipage}{0.18\textwidth}
				\begin{prooftree}
					\hypo{\phantom{X}}
					\infer[left label =\AxTraceLoop:]1{\phi \trace[2k] \phi, \Gamma}
				\end{prooftree}
			\end{minipage}
		\end{minipage}
		
		\ruleskip
		
		\begin{minipage}{\textwidth}
			\begin{minipage}{0.46\textwidth}
				\begin{prooftree}
					\hypo{\phi,\psi,\phi \land \psi \trace[1] \phi, \phi \land \psi \trace[1] \psi, \Gamma}
					\infer[left label =\RuAnd:]1{\phi \land \psi,\Gamma}
				\end{prooftree}
			\end{minipage}
			\begin{minipage}{0.50\textwidth}
				\begin{prooftree}
					\hypo{\phi, \phi \lor \psi\trace[1] \phi, \Gamma}
					\hypo{\psi,\phi \lor \psi \trace[1] \psi, \Gamma}
					\infer[left label =\RuOr:]2{\phi \lor \psi,\Gamma}
				\end{prooftree}
			\end{minipage}
		\end{minipage}
		
		\ruleskip
		
		\begin{minipage}{\textwidth}
			\begin{minipage}{0.43\textwidth}
				\begin{prooftree}
					\hypo{\phi[\eta x . \phi / x], \eta x. \phi \trace[\Omega(\eta x. \phi)] \phi[\eta x . \phi / x], \Gamma}
					\infer[left label =\RuEta:]1{\eta x . \phi, \Gamma}
				\end{prooftree}
			\end{minipage}
			\begin{minipage}{0.40\textwidth}
				\begin{prooftree}
					\hypo{\phi \trace[k] \psi, \psi \trace[l] \chi, \phi \trace[\max\{k,l\}] \chi, \Gamma}
					\infer[left label =\RuTrans:]1{\phi \trace[k] \psi, \psi \trace[l] \chi, \Gamma}
				\end{prooftree}
			\end{minipage}
			\begin{minipage}{0.1\textwidth}
				\begin{prooftree}
					\hypo{\Gamma}
					\infer[left label = \RuWeak:]1{A,\Gamma}
				\end{prooftree}
			\end{minipage}
		\end{minipage}
		
		\ruleskip
		
		\begin{minipage}{\textwidth}		
			\begin{minipage}{0.25\textwidth}
				\begin{prooftree}
					\hypo{\phi,\Sigma, \ldiap[\conv{a}]\Gamma, \Gamma^{\ldiap\phi}}
					\infer[left label =\RuDia:]1{\ldiap \phi, \lboxp \Sigma, \Gamma}
				\end{prooftree}
			\end{minipage}
			\begin{minipage}{0.34\textwidth}
				\begin{prooftree}
					\hypo{\phi, \Gamma}
					\hypo{\mybar{\phi}, \Gamma}
					\infer[left label =\RuCut:]2[~$\phi \in \ClosN(\Gamma)$]{\Gamma}
				\end{prooftree}
			\end{minipage}
			\begin{minipage}{0.37\textwidth}
				\begin{prooftree}
					\hypo{\phi \trace[k] \psi, \Gamma}
					\hypo{\phi \ntrace[k] \psi, \Gamma}
					\infer[left label =\RuTCut:]2[~$\phi \in \ClosN(\Gamma)$]{\Gamma}
				\end{prooftree}
			\end{minipage}
		\end{minipage}
	\end{mdframed}
	\caption{Proof rules of the tableau system \NWtwo}
	\label{fig.NWtwo}
\end{figure*}

\begin{definition}[Derivation]\label{def.NWtwoDerivation}
An \emph{\NWtwo derivation} $\pi = (T,P,\sfS,\sfR, \sff)$ is a proof tree 
defined from the rules in Figure \ref{fig.NWtwo} such that $(T,P)$ is a, 
possibly infinite, tree with nodes $T$ and parent relation $P$;
$\sfS$ is a function that maps each node $u \in T$ to a sequent 
$\sfS_u$; $\sfR$ is a function that maps each node $u \in T$ to the name of
a rule in Figure \ref{fig.NWtwo}; and $\sff: T \to \sfS_u \cup \{\nil\}$ is a
function that maps each node $u \in T$ to its \emph{principal formula} or to 
$\nil$ if the rule does not have a principal formula.		
Finally, we require the maps $\sfS, \sfR$ and $\sff$ to be in accordance with
the formulation of the rules.
\end{definition}

In our notation trees grow \emph{upwards}.
We call an \NWtwo derivation $\pi$ \emph{regular}, if it has finitely many distinct subderivations.

\begin{remark}
An occurrence of a rule is usually called \emph{analytic}, if all formulas $\phi$ in the premiss of the rule are subformulas of the conclusion $\Gamma$. In the context of fixpoint logics, this notion has to be extended, such that we demand that all formulas $\phi$ are in the \emph{closure} of $\Gamma$.
Because the rules \RuCut, \RuTCut and \RuDia are restricted, all rules in \NWtwo are analytic.
\end{remark}

Different than usual, we define the trail relation to consist of triples, where we include the weight of the trail. The weight keeps track of the priority of unfolded fixpoints along the trail. 

\begin{definition}[Trails]\label{def.NWtwoTrails}
	Let $\Gamma,\Gamma'$ be sequents and $\xi \in \Gamma \cup \{\nil\}$, such that $\Gamma$ describes the conclusion, and $\Gamma'$ a premiss of a rule $\Ru$ in Figure \ref{fig.NWtwo} and $\xi$ is either the principal formula of $\Ru$ or else $\Ru$ does not have a principal formula and $\xi = \nil$.
	
We define the \emph{upward trail relation} $\sfT_{\Gamma,\xi,\Gamma'} \subseteq
\Phi \times \Phi \times \Nat$ as follows by a case distinction on the principal 
formula $\xi$:
	\begin{itemize}
		\item  if $\xi = \ldiap \phi$, then $\sfT_{\Gamma,\xi,\Gamma'} \isdef \{(\ldiap \phi, \phi, 1)\} \cup \{(\lboxp \psi, \psi, 1) \| \lboxp \psi \in \Gamma\}$,
		\item  if $\xi = \eta x. \phi$, then $\sfT_{\Gamma,\xi,\Gamma'} \isdef \{(\eta x. \phi, \phi\subst{\eta x. \phi}{x}, \Omega(\eta x. \phi))\} \cup \{(\psi, \psi, 1) \| \psi \in \Gamma \cap \Gamma'\}$,
		\item if $\xi = \nil$, then $\sfT_{\Gamma,\xi,\Gamma'} \isdef \{(\psi, \psi, 1) \| \psi \in \Gamma \cap \Gamma'\}$,
		\item Similarly for the other cases.
	\end{itemize} 
In all of these cases we call $(\phi,\psi,k)$ an \emph{upward trail}.
	
Let $u,v$ be nodes in an \NWtwo derivation, such that $v$ is a child of $u$ or
$u = v$. 
We define the \emph{trail relation} $\sfT_{u,v} \subseteq \sfS_u \times \sfS_v 
\times \Nat$ as follows. If $u = v$ we define $\sfT_{u,u} \isdef \{(\phi, \psi, k) \| \phi \trace[k] \psi \in \sfS(u)\}$ and call $(\phi,\psi,k)$ a \emph{detour trail}. 
Otherwise $v$ is the child of $u$ and we define $\sfT_{u,v} \isdef \sfT_{\sfS(u),
\sff(u),\sfS(v)}$.
	
Let $\alpha$ be a branch in an \NWtwo proof $\pi$. 
A \emph{trail} on $\alpha$ is a sequence of upward trails with inserted detour trails. Due to the presence of cuts, trails do not necessarily start at the root, but could also start from a cut formula. Formally a \emph{trail} $\tau$ on $\alpha$ is a word in 
	\[ \sfT_{v_N,v_{N+1}}  (\sfT_{v_{N+1},v_{N+1}})^* \sfT_{v_{N+1},v_{N+2}} (\sfT_{v_{N+2},v_{N+2}})^*\cdots\] 
for some $N\geq 0$ and such that for any subword $(\phi, \psi,k)(\chi, \zeta,l)$ 
of $\tau$ it holds $\psi = \chi$. 
An infinite trail $\tau$ is called a \emph{$\mu$-trail}, if $\max\{k \| k\text{ appears infinitely often on } \tau\}$ is even and a \emph{$\nu$-trail} otherwise.
\end{definition}

\begin{definition}[Proof]\label{def.NWtwoProof}
	An \emph{\NWtwo proof} $\pi$ is an \NWtwo derivation such that on every infinite branch of $\pi$ there is a $\mu$-trail.
\end{definition}
We say that \NWtwo proves a sequent $\Gamma$, written $\NWtwo \proves \Gamma$ if there is an \NWtwo proof $\pi$, where the root is labelled by $\Gamma$.

\subsection{Proof search game}

A \emph{rule instance} is a triple $(\Delta,\Ru, \langle\Delta_1,...,\Delta_n\rangle)$ such that $\begin{prooftree}
	\hypo{\Delta_1}
	\hypo{\cdots}
	\hypo{\Delta_m}
	\infer[left label = \Ru]3[]{\Delta}
\end{prooftree}$
is a valid rule application in \NWtwo. We let $\conc$ be the function mapping rule instances to their conclusions.
A rule instance is \emph{cumulative} if all premisses are supersets of the conclusion and \emph{productive} if all the premisses are distinct from its conclusion.

We define the proof search game $\calG(\Phi)$, here we call the two players Prover and Builder. Its positions are given by $\Seq_{\Phi}\cup \Inst_{\Phi}$, where $\Seq_{\Phi}$ is the set of sequences and $\Inst_{\Phi}$ the set of rule instances containing only formulas in $\Phi$. The ownership function and admissible moves are given in Table \ref{tab.proofSearchGame}.
\begin{table}[htb]
	\begin{center}
	\begin{tabular}{|c|c|c|}
		\hline
		Position & Owner & Admissible moves \\
		\hline
		$\Gamma$ & Prover & $\{i \in \Inst_{\Phi} \| \conc(i) = \Gamma\}$ \\
		$(\Gamma,\sfR, \langle\Delta_1,\cdots,\Delta_m\rangle)$ & Builder & $\{\Delta_i \| i = 1,\cdots,m\}$ \\
		\hline
	\end{tabular}
\end{center}
\caption{The proof search game $\calG(\Phi)$}
\label{tab.proofSearchGame}
\end{table}
\vspace{-12pt}

An infinite match is won by Prover iff on the resulting branch there is a $\mu$-trail. 

An \NWtwo proof of $\Gamma$ can be seen as a winning strategy of Prover in
$\calG(\Phi)@\Gamma$. 
Note that $\calG(\Phi)$ is an $\omega$-regular game and that $\omega$-regular games have finite-memory strategies. Therefore we may assume that winning strategies for Prover are regular and consequently that \NWtwo proofs are regular.

\subsection{Soundness and Completeness}\label{sec.sub.SoundnessNWtwo}
For proving soundness we need to show that, if \NWtwo proves $\Gamma$, then $\Gamma$ is unsatisfiable. 
By contradiction we assume that $\Gamma$ is satisfiable 
and show that \NWtwo does not prove $\Gamma$. 
To do so, we assume a pointed model $\bbS,s$ and a strategy $f$ for $\eloi$ in 
$\calE(\Land \Phi, \bbS)$ such that $\bbS,s \sat_f \Gamma$. 
Using $f$ we can construct a winning strategy $\overline{f}$ for Builder in
$\calG(\Phi)@\Gamma$. The proof is rather straight-forward and can be found in  Appendix \ref{app.soundness}.
\begin{theorem}[Soundness]\label{thm.SoundnessNWtwo}
	If $\NWtwo \proves \Gamma$, then $\Gamma$ is unsatisfiable.
\end{theorem}

To prove completeness for pure sequents $\Gamma$, we follow the same proof strategy as in \cite{rood:focu23}. We only state the definitions and main ideas and refer to Appendix \ref{app.completeness} for full proofs.

We show that every unsatisfiable pure sequent is provable in \NWtwo. 
By contraposition, given a winning strategy $f$ for Builder in
$\calG(\Phi)@\Gamma$, we construct a model $\bbS^f$ and a positional strategy 
$\underline{f}$ for $\eloi$ in $\calE:=\calE(\Land \Phi, \bbS^f)$ such that 
$\bbS^f \sat_{\underline{f}} \Gamma$. 

Let $\calT$ be the subtree of the game-tree of $\calG(\Phi)@\Gamma$, where Builder plays the strategy $f$ and Prover picks rule instances according to the following priorities: 1) axioms \AxLit, \AxBot, \AxTraceNeg and \AxTraceLoop; 2) cumulative and productive instances of \RuOr, \RuAnd, \RuMu, \RuNu and 
\RuTrans; 3) cumulative and productive instances of \RuCut and \RuTCut; 4) modal rules \RuDia.

We say that a trace atom $\phi \trace[k] \psi$ is \emph{relevant} if 
(i) $\psi \in \Clos(\phi)$ and (ii) $\phi,\psi$ contain fixpoints. 
We restrict instances of \RuTCut and \RuDia to only introduce 
relevant trace atoms. 
For the rule \RuDia this amounts to changing the rule to a
variant, where only relevant trace atoms occur in the premiss. This rule can easily shown to be admissible using \RuWeak.
Because of these assumptions we may assume that all trace atoms in a constructed proof are relevant.

The model $\bbS^f= (S^f, \{R_a^f\}_{a \in \Act}, V^f)$ is defined as follows: The set $S^f$ of states consists of maximal upward paths $\rho$ in $\calT$ not containing a modal rule. We write $\sfS(\rho) = \bigcup\{\Delta \| \Delta \text{ occurs in } \rho\}$. We write $\rho_1 \to[a] \rho_2$ if $\rho_2$ is directly above $\rho_1$ only separated by an application of $\RuDia$. The relations $R_a^f$ are defined as follows:
\[\rho_1 R_a^f \rho_2  \quad :\Leftrightarrow \quad \rho_1 \to[a] \rho_2 \text{ or } \rho_2 \to[\conv{a}] \rho_1.\]
The valuation $V^f(p)$ is defined as $V^f(p) := \{\rho \| p \in \sfS(\rho)\}$.

The strategy $\underline{f}$ for $\eloi$ in  $\calE$ is defined as follows:
\begin{itemize}
\item At $(\phi_0 \lor \phi_1,\rho)$ pick a disjunct $\phi_i$ such that 
$\phi_i \in \sfS(\rho)$.
\item At $(\ldiap \phi, \rho)$ choose $(\phi,\tau)$ for some $\tau$ such that 
$\rho \overset{a}{\rightarrow} \tau$, where the principal formula in the rule instance \RuDia between $\rho$ and $\tau$ is 
$\ldiap \phi$.
\end{itemize}

Let $\rho_0$ be a state of $\bbS^f$ containing the root $\Gamma$ of $\calT$ and let $\phi_0 \in \Gamma$. We can show that the strategy $\underline{f}$ is well-defined and winning for $\eloi$ in $\calE@(\phi_0,\rho_0)$. From this completeness follows.

\begin{theorem}[Completeness]\label{thm.CompleteenessNWtwo}
	If a pure sequent $\Gamma$ is unsatisfiable, then $\Gamma$ is provable in \NWtwo.
\end{theorem}

\subsection{Tracking automaton for \NWtwo}
The main difficulty of working with \NWtwo is the handling of infinite trails, as they behave non-deterministically: trails split and merge. We define a non-deterministic $\epsilon$-parity automaton $\bbA_{\Phi}$, that checks the success condition on infinite paths and will be used later to show completeness of the \JStwo proof system. 
First we have to bring \NWtwo proofs in a certain normal form.

We call an \NWtwo proof $\pi$ \emph{saturated}, if (i) the rule \RuTrans is always applied when applicable and (ii) all applications of \RuEta rules are cumulative. Note that every \NWtwo proof $\pi$ can easily be transformed into a saturated proof $\pi'$ of the same sequent.

We call an infinite trail $\tau$ \emph{slim}, if (i) there are no two consecutive detour trails on $\tau$ and (ii) there is no upward trail of the form $(\eta x. \phi, \phi\subst{\eta x. \phi}{x},k)$ on $\tau$.

\begin{lemma}\label{lem.slimTrails}
	Let $\pi$ be a saturated $\NWtwo$ proof of $\Gamma$. On every infinite branch of $\pi$ there is a slim $\mu$-trail.
\end{lemma}

We will define a nondeterministic parity automaton $\bbA_{\Phi}$, called \emph{tracking automaton}, that checks if an infinite branch $\alpha$ of a saturated \NWtwo proof $\pi$ is successful. 
Conceptually the automaton $\bbA_{\Phi}$ non-deterministically follows the trail relation on $\alpha$. The states of the automaton will be formulas and trace atoms, where we add extra states for every fixpoint formula, in order to track the unfolding of fixpoints. Additionally we have an extra initial state, which is always reachable, as trails may start at any node.

Upward trails on $\alpha$ correspond to basic transitions in $\bbA_{\Phi}$ and detour trails are modelled by $\epsilon$-transitions going through a trace atom. 
In order to simplify the automaton we do not allow consecutive $\epsilon$-transitions and $\epsilon$-transitions starting from the auxiliary formula $\phi\subst{\eta x. \phi}{x}$ of the rule $\RuEta$. Hence $\bbA_{\Phi}$ will not follow all infinite trails, but only those of a simple form, in particular all slim trails. 

The alphabet $\Sigma$ consists of all triples $(\Gamma,\xi,\Gamma')$, where 
$\Gamma \subseteq \Phi$  describes the conclusion, and $\Gamma' \subseteq 
\Phi$ a premiss of a rule $\Ru$ in Figure \ref{fig.NWtwo} and $\xi$ is either the principal formula of $\Ru$ or else $\Ru$ does not have a principal formula and $\xi = \nil$.

We define the following nondeterministic $\epsilon$-parity automaton $\bbA_{\Phi} = 
(A,\Delta, a_I,\Omega_A)$:

\vspace{0.2cm}
{\centering
	$ \displaystyle
		\begin{aligned}
			A \isdef &\{a_I\} \cup \Phi \cup \{\eta x. \psi^* \| \eta x. \psi \in \Phi\} \\
	&\cup  \{\phi \trace[k] \psi \| \phi \trace[k] \psi \text{ a trace atom} \}. 
		\end{aligned}
	$
\par}
\vspace{0.2cm}

\noindent
For each $\gamma \in A $ and $(\Gamma,\xi,\Gamma') \in \Sigma$ we define $\Delta_b$ as follows.
	\begin{enumerate}
		\item if $\gamma = a_I$, then $\Delta_b(\gamma,(\Gamma,\xi,\Gamma')) \isdef \Gamma'\cup\{a_I\}$,
		\item if $\gamma = \xi = \eta x.\psi$ then $\Delta_b(\gamma,(\Gamma,\xi,\Gamma')) \isdef \{\eta x. \psi^*\}$,
		
\item else if  $\gamma = \phi \in \Phi$ then 
\vspace{-0.1cm}
\[\begin{array}{l}
\Delta_b(\phi,(\Gamma,\xi,\Gamma')) \isdef \{\phi' \| (\phi,\phi',1) \in \sfT_{\Gamma,\xi,\Gamma'}\} \cup 
\\ \{\phi'\trace[k] \psi' \| (\phi,\phi',1) \in \sfT_{\Gamma,\xi,\Gamma'}
\text{ \& }\psi',\phi'\trace[k] \psi' \in \Gamma'\}.
\end{array}\]

\item else $\Delta_b(\gamma,(\Gamma,\xi,\Gamma')) \isdef \nada$.
	\end{enumerate}
	For each $\gamma \in A$ we define $\Delta_{\epsilon}$ as follows. 
	\begin{enumerate}
		\item if $\gamma = \eta x. \psi^*$, then $\Delta_{\epsilon}(\eta x. \psi^*) \isdef= \{\eta x. \psi\}$,
		\item if $\gamma = \phi \trace[k] \psi$ then $\Delta_{\epsilon}(\gamma) \isdef \{\psi\}$ and
		\item else $\Delta_{\epsilon}(\gamma) \isdef \nada$.
	\end{enumerate}
	For states of the form $\eta x. \psi^*$ let $\Omega_A(\eta x. \psi^*) \isdef \Omega(\eta x. \psi)$. For states of the form $\phi \trace[k] \psi$ let $\Omega_A(\phi \trace[k] \psi) \isdef k$. For all other states $\gamma$ let $\Omega_A(\gamma) \isdef 1$.

Let $\alpha = (v_i)_{i\in\omega}$ be an infinite branch in an \NWtwo proof $\pi$. 
We define the stream $w(\alpha) \in \Sigma^{\omega}$ induced by $\alpha$ to be 
the infinite word $(\sfS(v_0),\sff(v_0),\sfS(v_1))(\sfS(v_1),\sff(v_1),\sfS(v_2))\cdots$. 

The following proposition states the adequacy of the tracking automaton.

\begin{proposition}\label{prop.trackAut}
	Let $\alpha$ be an infinite branch in a saturated \NWtwo proof. Then $\alpha$ is successful iff $w(\alpha) \in \calL(\bbA_{\Phi})$. 
\end{proposition}
\begin{proof}
	If $\alpha$ is successful, then $\alpha$ carries a slim $\mu$-trail due to Lemma \ref{lem.slimTrails} and therefore $w(\alpha) \in \calL(\bbA_{\Phi})$. Conversely,  $w(\alpha) \in \calL(\bbA_{\Phi})$ implies that there is a $\mu$-trail on $\alpha$.
\end{proof}

\section{Determinization of $\epsilon$-parity automata}\label{sec.automata}

In the last section the tracking automaton $\bbA_{\Phi}$ was defined, which 
checks if an infinite branch in an \NWtwo proof is successful. We want to find an equivalent deterministic automaton $\bbB$ and use it to obtain a simpler proof system. For this to work we have to define the determinization method in a simple way, in particular it has to rely on a ``powerset-construction'', meaning that states of $\bbB$ consist of subset of states in $\bbA_{\Phi}$ plus some extra information. Our method is a generalization of the Safra construction for Büchi automata \cite{Safra1988}.

Let $\Sigma$ be an alphabet and fix an $\epsilon$-parity automaton $\bbA = \langle A, \Delta, a_I, \Omega \rangle$. Let $n$ be the maximal even priority of $\bbA$ and $m = |A|$ be the size of $\bbA$.

\begin{definition}
Let $a \in A$. The \emph{$k$-parity $\epsilon$-Closure} of $a$, written $\eClos_k(a)$, consists of all states $b \in A$ for which there is a $\Delta_{\epsilon}$-path 
$a = a_0a_1\cdots a_n=b$ in $\bbA$ with $\max\{\Omega(a_i)\| i = 1,\ldots,n\} = k$.
\end{definition}

For each even number $k = 0,2,\ldots,n$ we fix a set of \emph{$k$-names}
$X_k$, such that $|X_k| = 4m$ and $X_k \cap X_l = \nada$ if $k \neq l$.
We define the set of \emph{names} $X = X_0 \uplus X_2 \uplus \cdots \uplus X_n$ and use the symbols $\dx,\dy,\dz,...$ for names. We call a non-repeating sequence of $k$-names $\tau_k$ a $k$-\emph{stack} and let $T_k$ be the set of all $k$-stacks. The empty sequence will be denoted by $\epsilon$. 
We define the set of all stacks $T$ to be $T_n \cdots T_2 \cdot T_0$, for clarity 
$T \isdef \{\tau_n\cdots\tau_2\cdot\tau_0 \mid \tau_n \in T_n,\ldots, \tau_0 \in T_0 \}$.
In case $\tau_{i} = \epsilon$ for all $i<k$ we may write
$\tau_n\cdots\tau_k$ rather than $\tau_n\cdots\tau_2\cdot\tau_0$.
For a stack $\tau$ we define $\tau \downto l$ to be the stack obtained from 
$\tau$ by removing all $k$-names, where $k < l$.

Each non-repeating sequence of names $\theta$ defines a \emph{linear order}
$<_{\theta}$ on names by setting $\dx <_\theta \dy$ if $\dx$ occurs before $\dy$ in $\theta$. 
This order extends to an order on stacks as follows: $\sigma <_\theta \tau$ if 
either
\begin{itemize}
\item $\sigma \downto k$ is a proper extension of $\tau \downto k$ for some 
$k \leq n$, or
\item $\sigma$ is lexicographically $<_\theta$-smaller than $\tau$, meaning that 
$\sigma$ and $\tau$ can be written as $\sigma = \rho \cdot \dx \cdot \sigma'$ and
$\tau = \rho \cdot \dy \cdot \tau'$ with $\dx <_\theta \dy$.
\end{itemize}
It can be proven that $<_\theta$ is a linear order on any set $T_0 \subseteq T$ such that $\theta$ contains all names in $T_0$.

Let $\theta, \tau$ be non-repeating sequences of names. We say that $\theta$ is an \emph{initial segment} of $\tau$ if $\tau = \theta \cdot \sigma$ for some sequence of names $\sigma$.
We say that $\theta$ is a \emph{subword} of $\tau$, notation: $\theta \sqsubseteq \tau$ if every name that occurs in $\theta$ also occurs in $\tau$ and if for all names $\dx, \dy$ in $\theta$ the relation $\dx <_{\theta} \dy$ implies $\dx <_\tau \dy$.

We now define the deterministic Rabin automaton $\bbA^S \isdef \langle A^S, \delta_A,
a'_I, R_A \rangle$.
Its carrier set $A^S$ consists of tuples $(A_0,f,\theta,c)$, where
\begin{itemize}
\item  $A_0$ is a subset of $A$,
\item $f: A_0 \rightarrow T$ maps each state $a$ of $A_0$ to a stack $\tau \in T$, such that $\tau = \tau\downto k$, where $k = \Omega(a)$,
\item $\theta$ is a non-repeating sequence of all names occurring in $\ran (f)$.
\item $c$ is a map $c: \ran(f) \to \{ \green, \white\}$.
\end{itemize}
A subset $A_0 \subseteq A$ will be called a \emph{macrostate} and we call $S_0
\in A^S$ a \emph{Safra-state}, in other words a Safra-state is a macrostate with additional
information. We will present a Safra-state $S_0 \in A^S$ as a set of elements 
$a^{\tau}$, where $a \in A$, $\tau \in T$ and $f(a) = \tau$, and deal with 
$\theta$ and $c$ implicitly. 
The sequence $\theta$ will be called the \emph{control}. We say a name is \emph{active} if it appears in $\theta$. An active $k$-name is \emph{visible} if it is the last $k$-name in some stack and \emph{invisible} otherwise. 
The function $c$ is called the \emph{colouring map} and we say that a name $\dx$ is coloured $\green$/$\white$, if $c(\dx) = \green$/$c(\dx) = \white$.

The initial Safra-state is $a'_I \isdef  \{a_I^{\epsilon}\}$. 
To define the transition function $\delta_A$ let $S$ be in $A^S$ and $z \in 
\Sigma$.  
We define $\delta_A(S,z) \isdef S'$, where $S'$ is constructed in the following steps.
Note that intermediate positions in this construction are not necessarily 
Safra-states; in particular there may be multiple stacks associated with some 
states.
\begin{enumerate}
\item \underline{Basic move:} For every $a^{\tau} \in S$ and $(a,z,b) \in \Delta_b$, add $b^{\tau\downto k}$ to $S'$, where $\Omega(b) = k$.
\item \underline{Cover:} For every $a^{\tau} \in S'$, where $\Omega(a) = k$ is even, change $a^{\tau}$ to $a^{\tau \cdot \dx}$, where $\dx$ is a fresh $k$-name that is not active in $S\cup S'$. If two different states are labelled by the same stack, we add the same name $\dx$. Add $\dx$ as the last element in $\theta$. 
\item \underline{$\epsilon$-Move:} For every $a^{\tau} \in S'$, odd $k$ and $b \in \eClos_k(a)$ add $b^{\tau \downto k}$ to $S'$. For every $a^{\tau} \in S'$, even $k$ and $b \in \eClos_k(a)$ add $b^{\tau \downto k \cdot \dx}$ to $S'$, where $\dx$ is a fresh $k$-name that is not active in $S \cup S'$. Add $\dx$ as the last element in $\theta$.
\item \underline{Thin:} 
If $a^{\sigma}$ and $a^{\tau}$ are in $S'$ and $\sigma <_\theta \tau$, remove $a^\tau$.
\item \underline{Reset:} 
Colour any invisible name $\dx$ green and change $a^{\sigma \cdot \dx \cdot \tau}$ to
$a^{\sigma \cdot \dx}$ for every  $a^{\sigma \cdot \dx \cdot \tau} \in S'$.
\end{enumerate}
Any name removed in this process is also removed from $\theta$.

The automaton $\bbA^S$ accepts a stream if in its run some name $\dx$ is active 
cofinitely often and coloured green infinitely often. 

\begin{remark}
In any Safra-state occurring in a run of $\bbA^S$ there are at most $m$ active $k$-names for every $k=0,2,...,n$, otherwise there is an invisible $k$-name and the Safra-state changes in step 5 of $\delta_A$. In step 2 of $\delta_A$ at most $m$ fresh $k$-names are introduced, resulting in at most $2m$ distinct elements $a^\tau$ in $S'$ after step 2. In step 3 up to $2m$ names are added. Thus in total at most $4m$ $k$-names are needed for each $k = 0,2,...n$.
\end{remark}

\begin{remark}
It may seem that the transition map $\delta_{A}$ is formulated in a 
non-deterministic way, but this is only superficially so: all choices can be 
made \emph{canonical}, based on an arbitrary but fixed order on names and 
states.
\end{remark}

\begin{theorem}\label{thm.determinisationCorrectness}
	The automaton $\bbA^S$ is equivalent to $\bbA$.
\end{theorem}
\begin{proof}
	The proof follows the same lines as the determinization for Büchi automata \cite{Safra1988} and can be found in Appendix \ref{app.determinisation}.
\end{proof}

\section{Annotated proof system}\label{sec.JS}
Following Jungteerapanich~\cite{Jungteerapanich2010} and 
Stirling~\cite{Stirling2014}, our approach is to use the determinization method 
of Section \ref{sec.automata} and build the automaton $\bbA_{\Phi}^S$ \emph{into} the proof 
system.
Hence, sequents of \JStwo correspond to Safra-states in $\bbA_{\Phi}^S$ and the rules 
of \JStwo correspond to the transition function in $\bbA_{\Phi}^S$.
This substantially simplifies the success condition on infinite paths and
allows us to formulate the cyclic proof system \JStwo.

\subsection{Cyclic \JStwo proofs}

Recall that we fixed a finite set of formulas $\Phi$ with associated priority function $\Omega: \Phi \to \Nat^+$ and let $n$ be the maximal even number in $\ran(\Omega)$.
As in Section \ref{sec.automata} let $\sfN_{k}$ be the set of \emph{$k$-names} for even $k = 0,2,\ldots,n$ and let $\sfN \isdef \bigcup\sfN_{k} $ be the \emph{set of names}. \emph{Stacks} are defined as in Section \ref{sec.automata}.


An \emph{annotated formula} is a pair $(\phi,\sigma)$, written as $\phi^\sigma$, where $\phi$ is a formula and $\sigma$ is a stack such that $\sigma = \sigma\downto \Omega(\phi)$. We will call $\sigma$ the \emph{annotation} of $\phi$. 
An \emph{annotated sequent} consists of a finite set of annotated formulas 
$\{\phi_1^{\sigma_1},\ldots,\phi_n^{\sigma_n}\}$, a set of trace atoms $T$ and a 
finite, non-repeating sequence of names $\theta$, called the \emph{control},
such that $\theta$ contains all names that occur in $\sigma_1,\ldots,\sigma_n$. 
The control can be seen as a linear order on the names occurring in a sequent; it
keeps track of when a name is added to a sequent.
If it is clear from the context we call annotated sequents just sequents and 
denote them as $\theta \proves \phi_1^{\sigma_1},\ldots,\phi_n^{\sigma_n}, T$. 
We use $A,B,\ldots$ as variables ranging over annotated formulas and trace atoms 
and use $\Gamma,\Delta, \Sigma,\ldots$ for sets of annotated formulas and trace 
atoms. 
For a set of formulas $\Gamma$ we define $\Gamma^{\epsilon} = \{\phi^\epsilon
\| \phi \in \Gamma\}$ and for an annotated sequent $\Gamma$ we define 
$\Gamma^{\epsilon} = \{\phi^{\epsilon}\| \phi^\sigma \in \Gamma \text{ for 
some } \sigma\}$.
Given an annotated sequent $\Gamma$ we define $\Clos(\Gamma) \isdef \Clos(\{\phi \in \muMLtwo \| \phi^\sigma \in \Gamma \text{ for some } \sigma\})$ and analogously for $\ClosN$.

In Figure \ref{fig.JStwo} the rules of the \JStwo proof system are given. If we ignore the control and the annotations, the axioms and the rules \RuAnd, \RuOr, \RuEta, \RuDia, \RuTrans, \RuWeak, \RuCut and \RuTCut coincide with the rules of \NWtwo.
Annotated sequents correspond to Safra-states of $\bbA_{\Phi}^S$, where $\bbA_{\Phi}$ is the tracking automaton checking the success condition on infinite \NWtwo paths. The transition function $\delta_A$ is split up between multiple rules: Step 1 is carried out in every rule; 
Step 2 adds a fresh name in \RuMu; 
Step 3 corresponds to the \RuJump rules; Step 4 is a special instance of \RuWeak and Step 5 corresponds to \RuReset.
We also add a weakening rule for names, called \RuExp. In order to obtain a cyclic system we add the discharge rule \RuDischarge{} that tracks repeats. Every \RuDischarge{} rule is labelled by a unique \emph{discharge token} taken from a fixed infinite set $\Tokens = \{\dd,\de,\ldots\}$.

\begin{figure*}[tbh]
	\begin{mdframed}[align=center]
		\begin{minipage}{\textwidth}
			\begin{minipage}{0.23\textwidth}
				\begin{prooftree}
					\hypo{\phantom{X}}
					\infer[left label =\AxLit:]1{\theta \proves \phi^\sigma, \mybar{\phi}^\tau, \Gamma}
				\end{prooftree}
			\end{minipage}
			\begin{minipage}{0.19\textwidth}
				\begin{prooftree}
					\hypo{\phantom{X}}
					\infer[left label =\AxBot:]1{\theta \proves \bot^{\sigma}, \Gamma}
					\end{prooftree}
			\end{minipage}
			\begin{minipage}{0.33\textwidth}
				\begin{prooftree}
					\hypo{\phantom{X}}
					\infer[left label =\AxTraceNeg:]1{\theta \proves \phi \trace[k] \psi, \phi \ntrace[k] \psi, \Gamma}
				\end{prooftree}
			\end{minipage}
			\begin{minipage}{0.18\textwidth}
				\begin{prooftree}
					\hypo{\phantom{X}}
					\infer[left label =\AxTraceLoop:]1{\theta \proves \phi \trace[2k] \phi, \Gamma}
				\end{prooftree}
			\end{minipage}
		\end{minipage}
		
		\ruleskip
		
		\begin{minipage}{\textwidth}
			\begin{minipage}{0.47\textwidth}
				\begin{prooftree}
					\hypo{\theta \proves\phi^\sigma,\psi^\sigma,\phi \land \psi \trace[1] \phi, \phi \land \psi \trace[1] \psi, \Gamma}
					\infer[left label =\RuAnd:]1{\theta \proves (\phi \land \psi)^\sigma,\Gamma}
				\end{prooftree}
			\end{minipage}
			\begin{minipage}{0.50\textwidth}
				\begin{prooftree}
					\hypo{\theta \proves \phi^\sigma, \phi \lor \psi\trace[1] \phi, \Gamma}
					\hypo{\theta \proves\psi^\sigma,\phi \lor \psi \trace[1] \psi, \Gamma}
					\infer[left label =\RuOr:]2{\theta \proves(\phi \lor \psi)^\sigma,\Gamma}
				\end{prooftree}
			\end{minipage}
		\end{minipage}
		
		\ruleskip
		
		\begin{minipage}{\textwidth}
			\begin{minipage}{0.5\textwidth}
				\begin{prooftree}
					\hypo{\theta\cdot \dx \proves\phi[\mu x . \phi / x]^{\sigma\downto k \cdot \dx}, \mu x. \phi \trace[k] \phi[\mu x.\phi / x], \Gamma}
					\infer[left label=\RuMu:]1[~$k = \Omega(\mu x. \phi)$ and $\dx$ is a fresh $k$-name]{\theta \proves\mu x . \phi^\sigma, \Gamma}
				\end{prooftree}
			\end{minipage}
		\end{minipage}
	
		\ruleskip
		
		\begin{minipage}{\textwidth}
			\begin{minipage}{0.60\textwidth}
				\begin{prooftree}
					\hypo{\theta \proves\phi[\nu x . \phi / x]^{\sigma\downto k}, \nu x. \phi \trace[k] \phi[\nu x . \phi / x], \Gamma}
					\infer[left label=\RuNu:]1[~$k = \Omega(\nu x. \phi)$]{\theta \proves\nu x . \phi^\sigma, \Gamma}
				\end{prooftree}
			\end{minipage}
			\begin{minipage}{0.35\textwidth}
				\begin{prooftree}
					\hypo{\theta \proves\phi^\sigma,\Sigma, \ldiap[\conv{a}]\Gamma^{\epsilon}, \Gamma^{\ldiap\phi}}
					\infer[left label =\RuDia:]1{\theta \proves\ldiap \phi^\sigma, \lboxp \Sigma, \Gamma}
				\end{prooftree}
			\end{minipage}
		\end{minipage}

		\ruleskip
		
		\begin{minipage}{\textwidth}
			\begin{minipage}{0.45\textwidth}
				\begin{prooftree}
					\hypo{\theta \proves\phi \trace[k] \psi, \psi \trace[l] \chi, \phi \trace[\max\{k,l\}] \chi, \Gamma}
					\infer[left label =\RuTrans:]1{\theta \proves\phi \trace[k] \psi, \psi \trace[l] \chi, \Gamma}
				\end{prooftree}
			\end{minipage}
			\begin{minipage}{0.2\textwidth}
				\begin{prooftree}
					\hypo{\theta \proves\Gamma}
					\infer[left label =\RuWeak:]1{\theta \proves A, \Gamma}
				\end{prooftree}
			\end{minipage}
			\begin{minipage}{0.20\textwidth}
			\begin{prooftree}
				\hypo{\theta' \proves \phi^{\tau},\Gamma}
				\infer[left label=\RuExp:]1[~$\theta' \sqsubseteq \theta$ and $\tau \sqsubseteq \sigma$]{\theta \proves \phi^\sigma, \Gamma}
			\end{prooftree}
			\end{minipage}
		\end{minipage}
			
		\ruleskip
				
		\begin{minipage}{\textwidth}		
			\begin{minipage}{0.42\textwidth}
				\begin{prooftree}
					\hypo{\theta\proves\phi^\sigma, \psi^{\sigma\downto2k+1}, \psi^\tau, \phi \trace[2k+1] \psi, \Gamma}
					\infer[left label=$\RuJump_o$:]1[]{\theta\proves\phi^\sigma, \psi^\tau, \phi \trace[2k+1] \psi, \Gamma}
				\end{prooftree}
			\end{minipage}
			\begin{minipage}{0.47\textwidth}
				\begin{prooftree}
					\hypo{\theta\cdot \dx\proves\phi^\sigma, \psi^{\sigma\downto2k\cdot \dx}, \psi^\tau, \phi \trace[2k] \psi, \Gamma}
					\infer[left label=$\RuJump_e$:]1[~$\dx$ is a fresh $2k$-name]{\theta\proves\phi^\sigma, \psi^\tau, \phi \trace[2k] \psi, \Gamma}
				\end{prooftree}
			\end{minipage}
		\end{minipage}
		\ruleskip

		\begin{minipage}{\textwidth}		
			
			\begin{minipage}{0.41\textwidth}
				\begin{prooftree}
					\hypo{\theta \proves\phi^\epsilon, \Gamma}
					\hypo{\theta \proves\mybar{\phi}^\epsilon, \Gamma}
					\infer[left label =\RuCut:]2[~$\phi \in \ClosN(\Gamma)$]{\theta \proves\Gamma}
				\end{prooftree}
			\end{minipage}
			\begin{minipage}{0.44\textwidth}
				\begin{prooftree}
					\hypo{\theta \proves\phi \trace[k] \psi, \Gamma}
					\hypo{\theta \proves\phi \ntrace[k] \psi, \Gamma}
					\infer[left label =\RuTCut:]2[~$\phi \in \ClosN(\Gamma)$]{\theta \proves\Gamma}
				\end{prooftree}
			\end{minipage}
			
		\end{minipage}
	
		\ruleskip
		
		\begin{minipage}{\textwidth}
			\begin{minipage}{0.68\textwidth}
				\begin{prooftree}
					\hypo{\theta \proves \phi_1^{\sigma\dx},\ldots, \phi_n^{\sigma\dx}, \Gamma}
					\infer[left label=\RuReset[\dx]:]1[~$\dx,\dx_1,\ldots,\dx_n$ are $k$-names, $\dx$ not in $\Gamma$]{\theta \proves \phi_1^{\sigma\dx\dx_1\tau_1},\ldots, \phi_n^{\sigma\dx\dx_n\tau_n}, \Gamma}
				\end{prooftree}
			\end{minipage}
			\begin{minipage}{0.17\textwidth}
				\begin{prooftree}
					\hypo{\discharge{\theta\proves\Gamma}{\dd}}
					\infer[no rule]1{\vdots}
					\infer[no rule]1{\theta\proves\Gamma}
					\infer[left label =\RuDischarge{\dd}:]1{\theta\proves\Gamma}
				\end{prooftree}
			\end{minipage}
		\end{minipage}
	\end{mdframed}
	\caption{Proof rules of the proof system \JStwo}
	\label{fig.JStwo}
\end{figure*}

\begin{definition}[Derivation]\label{def.JStwoDerivation}
	A \emph{\JStwo derivation} $\pi = (T,P,\sfS,\sfR, \sff)$ is a proof tree defined from the rules in Figure \ref{fig.JStwo} such that $T,P$ and $\sff$ are defined as for \NWtwo derivations; $\sfS$ maps each node $u \in T$ to an annotated sequent $\sfS_u$; and
	$\sfR$ is a function that maps every node $u \in T$ to either (i) the name of a rule in Figure \ref{fig.JStwo}, (ii)  a discharge token or (iii) an extra value $o$, such that every node labelled with a discharge token or $o$ is a leaf.
	
	For every leaf $l$ that is labelled with a discharge token $\dd \in \Tokens$ there is a proper ancestor $c(l)$ of $l$ that is labeled with \RuDischarge{\dd} and such that $l$ and $c(l)$ are labelled by the same sequent. In this situation we call $l$ a \emph{repeat leaf} and $c(l)$ its \emph{companion}.
	Leaves labelled by $o$ are called \emph{open assumptions}.
	
	Given $\pi$ we define the  usual \emph{proof tree} $\calT_\pi = (T,P)$ and the \emph{proof tree with back edges} $\cyclicPT = (T,P^C)$ where $P^C = P \cup \{(l,c(l))\mid l \text{ is a repeat leaf}\}$.
\end{definition}

\begin{definition}[Successful path]
	A finite path $\alpha$ in a \JStwo derivation is called \emph{successful} if there is a name $\dx$ such that
	\begin{enumerate}
		\item $\dx$ occurs in the control of every sequent on $\alpha$ and
		\item there is an application of $\RuReset[\dx]$ on $\alpha$.
	\end{enumerate}
	
Let $v$ be a  repeat leaf in a \JStwo derivation  $\pi = (T,P,\sfS,\sfR, \sff)$
with companion  $c(v)$, and let $\alpha_v$ denote the \emph{repeat path} of $v$ 
in $\calT_{\pi}$ from $c(v)$ to $v$. 
We say that $v$ is \emph{discharged} if the path $\alpha_v$ is successful.		
A leaf is called \emph{closed} if it is either discharged or labelled by an 
axiom, and is called \emph{open} otherwise.
\end{definition}

\begin{definition}[Cyclic proof]
	A \JStwo proof is a finite \JStwo derivation, where every leaf is closed.
\end{definition}
We say that \JStwo proves a set of formulas $\Gamma$, written $\JStwo \proves \Gamma$, if there is a \JStwo proof $\pi$, where the root is labelled by $\epsilon \proves \Gamma^{\epsilon}$.
\subsection{Infinite \JStwo proofs}

\begin{definition}
	Let $\pi$ be a \JStwo derivation and $\alpha$ an infinite path in $\pi$. We call $\alpha$ \emph{successful}, if there is a name $\dx$ such that 
	\begin{enumerate}
		\item $\dx$ occurs in the control of cofinitely many sequents on $\alpha$ and
		\item there are infinitely many applications of $\RuReset[\dx]$ on $\alpha$.
	\end{enumerate}
\end{definition}

\begin{definition}[Infinitary proofs]
	A \JStwoInfty proof is a \JStwo derivation, where every leaf is labelled by an axiom and every infinite path is successful.
\end{definition}

\begin{lemma}\label{lem.InftyJSiffJS}
	$\JStwo \proves \Gamma$ iff there is a regular $\JStwoInfty$ proof of $\Gamma$.
\end{lemma}
\begin{proof}
	First let $\pi$ be a $\JStwo$ proof of $\Gamma$. The infinite unfolding $\pi^*$ of $\pi$ is the $\JStwo$ derivation obtained from $\pi$ by recursively replacing every discharged leaf $l$ with the subtree of $\pi$ rooted at the child node of $c(l)$. It easy to see that $\pi^*$ is a regular $\JStwoInfty$ proof of $\Gamma$.
	
Conversely, let $\rho$ be a regular $\JStwoInfty$ proof. 
For a node $v \in \rho$ let $\rho_v$ be the subtree of $\rho$ rooted at $v$. 
For every infinite path $\alpha = (\alpha(i))_{i \in \omega}$ define minimal 
indices $j < k$ such that
	\begin{enumerate}
		\item $\rho_{\alpha(j)} = \rho_{\alpha(k)}$ and
		\item the path $\alpha(j)\cdots \alpha(k)$ is successful.
	\end{enumerate}
	Because $\rho$ is regular and every infinite path is successful, such indices always exist. For each such infinite path we introduce a \RuDischarge{\dd} node at $\alpha(j)$ and let $\alpha(k)$ be a leaf discharged by $\dd$. Using König's Lemma we can show that this procedure results in a finite \JStwo proof $\pi$ of $\Gamma$.
\end{proof}

\subsection{Soundness and Completeness}

The proof system \JStwoInfty was constructed as follows: Take an \NWtwo proof and define the tracking automaton $\bbA_{\Phi}$ that checks whether an infinite branch carries a $\mu$-trail. Using the determinization method from
Section \ref{sec.automata} we simulate Safra-states of $\bbA_{\Phi}^S$ by annotated sequents in the \JStwoInfty system. Thus Safra-states in $\bbA_{\Phi}^S$ correspond to annotated sequents and the transition function of $\bbA_{\Phi}^S$ corresponds to various rules of \JStwo. In particular, step 4 corresponds to a specific shape of \RuWeak, which we call \RuThin. We also need a particular instance of \RuExp, that only removes names from $\theta$ which do not occur in $\Gamma$:
\[\begin{minipage}{\textwidth} 
	\begin{minipage}{0.01\textwidth}
		\phantom{x}
	\end{minipage}
	\begin{minipage}{0.27\textwidth}
		\begin{prooftree}
			\hypo{\theta \proves \phi^\sigma, \Gamma}
			\infer[left label= \RuThin:]1[~$\sigma <_{\theta} \tau$]{\theta \proves \phi^\sigma, \phi^\tau, \Gamma}
		\end{prooftree}
	\end{minipage}
	\begin{minipage}{0.14\textwidth}
		\begin{prooftree}
			\hypo{\theta' \proves \Gamma}
			\infer[left label= $\RuExp'$:]1[~$\theta' \sqsubseteq \theta$]{\theta \proves \Gamma}
		\end{prooftree}
	\end{minipage}
\end{minipage} 
\]
Infinite runs of $\bbA_{\Phi}^S$ correspond to infinite branches in \JStwoInfty. This will be formalized in the proof of Lemma \ref{lem.JSinftyIffNW}.

\begin{lemma}\label{lem.JSinftyIffNW}
	There is a $\JStwoInfty$ proof $\rho$ of $\Gamma$ iff there is an $\NWtwo$ proof $\pi$ of $\Gamma$. The proof $\rho$ is regular iff $\pi$ is so.
\end{lemma}
\begin{proof}
	First let $\pi$ be an \NWtwo proof of a sequent $\Gamma$. We may assume that $\pi$ is saturated, otherwise add \RuTrans rules whenever applicable and make \RuEta rules cumulative. 
	Inductively we translate every node $v$ in $\pi$ to a node $v'$ (potentially with additional nodes), such that $v'$ is labelled by the same sequent as $v$ plus 
	annotations, where $v'$ corresponds to a macrostate in $\bbA_{\Phi}^S$.
	This can be achieved by replacing every rule in \NWtwo by its corresponding rule in 
	\JStwo and adding productive instances of the rules \RuJump, \RuThin, \RuReset[] and $\RuExp'$ whenever applicable (in that order bottom-up). 
	This yields a \JStwoInfty derivation $\rho$ that is regular if $\pi$ is regular. 
	It remains to show that every infinite branch $\alpha= (v_i)_{i\in\omega}$ in 
	$\rho$ is successful. 
	Let $\hat{\alpha}$ be the corresponding infinite branch in $\pi$. 
	Due to Proposition \ref{prop.trackAut} it holds that $w(\hat{\alpha}) \in 
	\calL(\bbA_{\Phi})$ and Theorem \ref{thm.determinisationCorrectness} yields $w(\hat{\alpha}) \in \calL(\bbA_{\Phi}^S)$.
	As the branch $\alpha$ closely resembles the run of $\bbA_{\Phi}^S$ 
	on $w(\hat{\alpha})$ it follows that $\alpha$ is successful.
	
Conversely let $\rho$ be a \JStwoInfty proof of $\Gamma$. 
We let $\pi$ be the \NWtwo derivation defined from $\rho$ by omitting the rules 
\RuExp, \RuJump and \RuReset[] and reducing all other rules to their corresponding 
\NWtwo rules by removing annotations. 
If $\rho$ is regular, so is $\pi$.
To show that $\pi$ is actually a proof, take an arbitrary branch $\alpha = (\alpha_i)_{i \in \omega}$; we 
have to prove that $\alpha$ is successful.

Let $\beta= (\beta_j)_{j \in \omega}$ be the corresponding infinite branch in 
$\rho$. 
In this direction we can not apply the determinization of the tracking automaton
directly, as in $\rho$ the rules do not have to be applied in a specific order,
meaning that branches in $\rho$ do not necessarily correspond to runs in 
$\bbA_{\Phi}^S$. 
Yet we show how one can reuse the proof of $\calL(\bbA_{\Phi}^S) \subseteq \calL(\bbA_{\Phi})$ 
(Converse direction of Theorem \ref{thm.determinisationCorrectness}) with only 
minor adaptions, here is the resulting proof sketch. 
As $\beta$ is successful, there is a $k$-name $\dx$ that occurs in the control
of cofinitely many sequents on $\beta$ and such that there are infinitely many 
applications of $\RuReset[\dx]$ on $\beta$. 
We can define minimal indices $t(0) < t(1) < \cdots$ such that $\dx$ occurs in 
the control of $\beta_{j}$ for $j\geq t(0)$ and such that in $\beta_{t(i)}$ the
rule $\RuReset[\dx]$ is applied for $i \in \omega$. 
The nodes $\beta_{t(i)}$ correspond to nodes $\alpha_{s(i)}$ on $\alpha$ for 
$i \in \omega$. 
As in the proof of Theorem \ref{thm.determinisationCorrectness} we can find 
trails $\tau_i$ from $\alpha_{s(i)}$ to $\alpha_{s(i+1)}$ with maximal weight
$k$. 
Using König's Lemma we can again glue together such trails and obtain an 
infinite $\mu$-trail on $\alpha$, which means that $\alpha$ is successful 
indeed.	
\end{proof}

\begin{theorem}[Soundness and Completeness]\label{thm.JStwoSoundnessCompleteness}
	A pure sequent $\Gamma$ is unsatisfiable iff $\JStwo \proves \Gamma$.
\end{theorem}
\begin{proof}
	From Theorem \ref{thm.SoundnessNWtwo} and Theorem \ref{thm.CompleteenessNWtwo} it follows that $\Gamma$ is unsatisfiable iff there is a regular \NWtwo proof of $\Gamma$. This is equivalent to the existence of a regular \JStwoInfty proof due to Lemma \ref{lem.JSinftyIffNW}. Hence Lemma \ref{lem.InftyJSiffJS} concludes the proof.
\end{proof}

\section{Split \JStwo system}\label{sec.splitJS}
Our overall strategy to prove interpolation is as follows: Given a \JStwo proof $\pi$ of $\phi, \psi$ we define a formula $I$ in the common vocabulary of $\phi$ and $\psi$ and construct proofs $\pi^l$ of $\phi, I$ and $\pi^r$ of $\mybar{I},\psi$. This is done by structural induction on $\pi$, where roughly $\pi^l$ contains those rules of $\pi$ concerning descendants of $\phi$ and $\pi^r$ contains those rules of $\pi$ concerning descendants $\psi$. In order to make that formal, we have to separate, in every sequent, those parts originating from $\phi$ and those originating from $\psi$. Sequents of this kind will be called split sequents.

\subsection{Split \JStwo proofs}
A \emph{split sequent} is a quadruple $(\theta, \Gamma, \kappa, \Delta)$, usually written 
as $\theta \proves \Gamma \| \kappa \proves \Delta$, where $\theta \proves \Gamma$ and $\kappa \proves \Delta$ are annotated sequents and $\theta$ and $\kappa$ are disjoint.
Note that we do not require that $\Gamma$ and $\Delta$ are disjoint. 
Given a split sequent $\theta \proves \Gamma \| \kappa \proves \Delta$ we call $\theta \proves \Gamma$ the left and $\kappa \proves \Delta$ the right component of the split sequent. 
We will write $\Psi^l$ and $\Psi^r$ for the left and right component of the split sequent $\Psi$, respectively.

We will define \emph{split \JStwo proofs} consisting of split sequents, where \JStwo rules are applied to either the left or the right component of a split sequent. Importantly, if $\Psi^l$ is the left component of the conclusion, all formulas in the left component of a premiss will be in $\ClosN(\Psi^l)$.

For any \JStwo rule \Ru we define a \emph{left \JStwo rule} $\Ru^l$ as follows. 
If $\Ru \neq \RuDia$ is of the form\footnote{For simplicity  we only depict a unary rule, the case of a binary rule is analogous.}
$$\begin{prooftree}
	\hypo{\theta' \proves \Gamma'}
	\infer[left label=\Ru:]1[]{\theta \proves \Gamma}
\end{prooftree}$$
then $\Ru^l$ is of the form 
$$\begin{prooftree}
	\hypo{\theta' \proves \Gamma' \| \kappa \proves \Delta}
	\infer[left label=$\Ru^l$:]1[]{\theta \proves \Gamma \| \kappa \proves \Delta}
\end{prooftree}$$
The rule $\RuDiaL$ is of the form 
	\[\begin{prooftree}
		\hypo{
		\theta \proves\phi^\sigma,\Sigma, \ldiap[\conv{a}]^l\Gamma^{\epsilon}, \Gamma^{\ldiap^l\phi}\| 
		\kappa \proves \Pi,\ldiap[\conv{a}]^r\Delta^{\epsilon},\Delta^{\ldiap^r\phi}}
		\infer[left label= \RuDiaL:]1[]{
		\theta \proves\ldiap \phi^\sigma, \lboxp \Sigma, \Gamma\| 
		\kappa \proves \lboxp \Pi,\Delta}
	\end{prooftree}\]

 Let $\Psi^l$ and $\Psi^r$ be the respective left and right component of the split sequent of the conclusion. Then we define $$\ldiap[\conv{a}]^l\Gamma \isdef \{\ldiap[\conv{a}]\gamma^\sigma \| \ldiap[\conv{a}]\gamma \in \ClosN(\Psi^l), \gamma^\sigma \in \Gamma\}.$$ The conditions in $\Gamma^{\ldiap\phi}$ are adapted, such that $\Gamma^{\ldiap^l \phi}$ is defined as 
\begin{align*}
	& ~\{\phi \ntrace[k]  \lboxp[\conv{a}] \chi \| \ldiap \phi \ntrace[k]  \chi \in \Gamma \text{ and } \lboxp[\conv{a}]\chi \in \ClosN(\Psi^l)\} \\
	\cup& ~\{\lboxp[\conv{a}] \chi \trace[k]  \phi \| \chi \trace[k]  \ldiap \phi \in \Gamma \text{ and } \lboxp[\conv{a}]\chi \in \ClosN(\Psi^l)\} \\
	\cup& ~\{\psi \ntrace[k]  \lboxp[\conv{a}] \chi \| \lboxp \psi \ntrace[k]  \chi \in \Gamma \text{ and } \lboxp[\conv{a}]\chi \in \ClosN(\Psi^l)\} \\
	\cup& ~\{\lboxp[\conv{a}] \chi \trace[k]  \psi \| \chi \trace[k]  \lboxp \psi \in \Gamma \text{ and } \lboxp[\conv{a}]\chi \in \ClosN(\Psi^l)\} 
\end{align*}
Analogously for $\ldiap[\conv{a}]^r\Delta$ and $\Delta^{\ldiap^r\phi}$.

\emph{Right \JStwo rules} are defined analogously.
Additionally we also allow so-called \emph{split axioms} of the form
\[
\begin{prooftree}
	\hypo{}
	\infer[left label=$\AxLit'$:]1{\theta \proves \phi^\sigma \| \kappa \proves \mybar{\phi}^\tau, \Delta}
\end{prooftree}\]

For most rules the left and the right component of the split do not interact.  The only exceptions are the modal rule \RuDia and the axiom $\AxLit'$. Note that for trace atoms there is no interaction between the left and the right component at all, and even the axiom \AxTraceNeg may only be applied if both a trace atom and its negated trace atom occur in the same component.

\begin{definition}
	A \emph{split \JStwo derivation} is a proof tree defined from all left and right \JStwo rules and split axioms.
\end{definition}

A finite path $\alpha$ in a \JStwo derivation is called \emph{left-successful} if there is a name $\dx$ such that
	\begin{enumerate}
		\item $\dx$ occurs in the left component of every split sequent on $\alpha$ and
		\item there is an application of $\RuReset[]^l_{\dx}$ on $\alpha$.
	\end{enumerate}
	
Let $v$ be a repeat leaf in a split \JStwo derivation $\pi$ with companion $c(v)$
labelled by $\RuDischarge{}^l$, meaning that $v$ is a descendant of $c(v)$ and 
$v$ and $c(v)$ are labelled by the same split sequent. 
Let $\tau_v$ denote the {repeat path} of $v$ from $c(v)$ to $v$. 
We say that the leaf $v$ is \emph{discharged} by $\RuDischarge{}^l$ if the path
$\tau_v$ is left-successful.	
Right-successful paths and leaves discharged by $\RuDischarge{}^r$ are defined 
analogously.
A leaf is called \emph{closed} if it is either discharged or labelled by an 
axiom, and is called \emph{open} otherwise.

\begin{definition}
	A \emph{split \JStwo proof} is a finite split \JStwo derivation, where every leaf is closed.
\end{definition}

Given sets of formulas $\Gamma$ and $\Delta$ we say that there is a split $\JStwo$ proof of $\Gamma \| \Delta$, if there is a split \JStwo proof of which the root is labelled by $\epsilon \proves \Gamma^{\epsilon} \| \epsilon \proves \Delta^{\epsilon}$.

\subsection{Soundness and completeness of split proofs}

For proving soundness of the split system we aim to translate a 
split \JStwo proof $\pi$ to a \JStwoInfty proof $\rho$. 
To do so we will translate split sequents of the form $\theta \proves \Sigma
\| \kappa \proves \Pi$ to sequents of the form $\zeta \proves \Sigma,\Pi$ 
where $\zeta$ is a combination of the names in $\theta$ and $\kappa$ according
to the following definition. 

Let $\theta$ and $\kappa$ be disjoint non-repeating sequences of names. We call a non-repeating sequence of names $\zeta$ a \emph{merger} of $\theta$ and $\kappa$, if $\zeta$ consists of exactly those names occurring in $\theta$ and $\kappa$ and such that $\theta \sqsubseteq \zeta$ and $\kappa \sqsubseteq \zeta$. 
Recall that we write $\theta \sqsubseteq \zeta$ if $\zeta$ contains all names occurring in $\theta$ and whenever a name $\dx$ occurs to the left of a name $\dy$ in $\theta$, then $\dx$ also occurs to the left of $\dy$ in $\zeta$.

\begin{lemma}\label{lem.splitJStwoSoundness}
	If there is a split \JStwo proof of $\Gamma \| \Delta$, then there is a regular \JStwoInfty proof of $\Gamma, \Delta$.
\end{lemma}
\begin{proof}
	Let $\pi$ be a split \JStwo proof of $\epsilon \proves \Gamma \| \epsilon \proves \Delta$.  Let $\pi^*$ be the infinite unfolding of $\pi$ -- the split \JStwo derivation obtained from $\pi$ by recursively replacing every discharged leaf $l$ with the subtree of $\pi$ rooted at the child node of $c(l)$.

	We inductively translate $\pi^*$ to a \JStwo derivation $\rho$ of $\Gamma,\Delta$, such that every node labelled by $\theta \proves \Sigma \| \kappa \proves \Pi$ in $\pi^*$ is translated to a node in $\rho $ labelled by $\zeta \proves \Sigma, \Pi$, where $\zeta$ is a merger of $\theta$ and $\kappa$. 
	
	This can be achieved by translating all rules of the form $\Ru^l$ and $\Ru^r$ to the corresponding rule $\Ru$ and split axioms to axioms \AxLit.
	It remains to show that $\rho$ is a \JStwoInfty proof, in other words, that every infinite path $\alpha$ in $\rho$ is successful. 
	Because $\pi$ is a split \JStwo proof, on every infinite path $\beta$ in $\pi^*$ there is a name $\dx$ and $d = l,r$ such that 
	\begin{enumerate}
		\item $\dx$ occurs in the control of $\Sigma^d$ on cofinitely many sequents $\Sigma$ on $\beta$ and
		\item there are infinitely many applications of $\RuReset[]_\dx^d$ on $\beta$.
	\end{enumerate} 
Therefore the analogous definition holds for every infinite path $\alpha$ in $\rho$ and thus $\rho$ is a $\JStwoInfty$ proof of $\Gamma,\Delta$. By construction $\rho$ is regular.
\end{proof}

In the soundness proof it sufficed to translate split \JStwo proofs to \JStwoInfty
proofs. The converse translation from \JStwoInfty proofs to split \JStwo proofs is more tricky, as we have to choose in which component formulas are put. 

\begin{lemma}\label{lem.JStwoSplitCompleteness}
If there is a regular \JStwoInfty proof of a pure sequent $\Gamma,\Delta$, then there is a split \JStwo proof of $\Gamma \| \Delta$.
\end{lemma}
\begin{proof}
Let $\rho$ be a \JStwoInfty proof of $\Gamma, \Delta$. 
We first translate $\rho$ to a split \JStwo derivation $\pi$ of $\Gamma \| \Delta$.
For the time being we assume that the bound variables in $\Gamma$ and $\Delta$ are disjoint.
Therefore all formulas in $\ClosN(\Gamma) \cap \ClosN(\Delta)$ are fixpoint-free.

In the completeness proof of \NWtwo an \NWtwo proof $\rho'$ was constructed such 
that all trace atoms  $\phi \trace[k] \psi$ in $\rho'$ are relevant, meaning that 
(i) $\psi \in \ClosN(\phi)$ and (ii) $\phi$ and $\psi$ contain fixpoints. 
In the completeness proof of \JStwo we translated $\rho'$ to a \JStwo proof
$\rho$ without adding extra trace atoms. 
Thus we may assume for every trace atom $\phi \trace[k] \psi$ in $\rho$ that $\phi,\psi \notin \ClosN(\Gamma) \cap \ClosN(\Delta)$ and either $\phi,\psi \in \ClosN(\Gamma)$ or $\phi,\psi \in 
\ClosN(\Delta)$. 
For simplicity, we write $\phi \trace[k] \psi \in \ClosN(\Sigma)$ in the case
that $\phi,\psi \in \ClosN(\Sigma)$.
	
We inductively translate $\rho$ to a split \JStwo derivation $\pi$ of 
$\Gamma \| \Delta$, such that every node $u$ labelled by $\theta \proves \Sigma$
in $\rho$ is translated to a node $v$ (possibly with some additional nodes) in $\pi$ labelled by $\theta^l \proves \Sigma^l \| \theta^r \proves \Sigma^r$ such that 
	\begin{enumerate}
		\item $\Sigma = \Sigma^l \cup \Sigma^r$, 
		\item $\Sigma^l \subseteq \ClosN(\Gamma)$ and $\Sigma^r \subseteq \ClosN(\Delta)$,
		\item $\Sigma^l \cap \Sigma^r = \nada$,
		\item $\theta^l \sqsubseteq \theta$ consists of exactly all names occurring in $\Sigma^l$,
		\item $\theta^r \sqsubseteq \theta$ consists of exactly all names occurring in $\Sigma^r$.
	\end{enumerate}
	
The root $\epsilon \proves \Gamma, \Delta$ is translated to \[
\begin{prooftree}
	\hypo{\epsilon \proves 
		\Gamma \| \epsilon \proves \Delta\setminus\Gamma}
	\infer[left label= $\RuWeak^r$]1{\epsilon \proves \Gamma \| \epsilon \proves \Delta}
\end{prooftree}\]
For every rule in $\rho$ we apply a corresponding left or right rule in $\pi$. By a case distinction on the applied rule we show how to satisfy conditions 1 and 2.
	\begin{itemize}
		\item \AxLit can either be translated to $\AxLit^l$, $\AxLit^r$ or to a split axiom $\AxLit'$, depending on in which components the formulas $\phi$ and $\mybar{\phi}$ are located.
		
		\item Assume that in $\rho$ the following \RuDia rule is applied:
		\[\begin{prooftree}
			\hypo{\theta \proves\phi^\sigma,\Sigma, \ldiap[\conv{a}]\Lambda^{\epsilon}, \Lambda^{\ldiap\phi},\Pi,\ldiap[\conv{a}]\Theta^{\epsilon},\Theta^{\ldiap\phi}}
			\infer[left label= $\mathsf{R}_{\ldiap}$]1[]{\theta \proves\ldiap \phi^\sigma, \lboxp \Sigma, \Lambda, \lboxp \Pi,\Theta}
		\end{prooftree}\]
		 Let the split of the translation of the conclusion in $\pi$ be 
		 \[ \Psi = \theta^l \proves \ldiap \phi^\sigma, \lboxp \Sigma, \Lambda \| \theta^r \proves \lboxp \Pi,\Theta.\]
		 Let $\Psi^l$ be the left, and $\Psi^r$ be the right component of $\Psi$. If we just try to apply $\Ru_{\ldiap}^l$ to $\Psi$ this will not work: It could be that there is $\gamma^{\tau} \in \Lambda\setminus \Theta$ and $\ldiap[\conv{a}] \gamma^\epsilon \in \ldiap[\conv{a}] \Lambda^{\epsilon}$ such that $\ldiap[\conv{a}] \gamma \in \ClosN(\Psi)$ but $\ldiap[\conv{a}] \gamma \notin \ClosN(\Psi^l)$. Hence, $\ldiap[\conv{a}] \gamma^\epsilon$ would be added neither in the left nor the right component of the premiss of $\Ru_{\ldiap}^l$, yet in $\rho$ the formula is added to the premiss of $\RuDia$.
				 
		In this case we must have $\ldiap[\conv{a}] \gamma \in \ClosN(\Psi^r)$. 
		Thus $\gamma \in \ClosN(\Psi^r)$ as well, and thence $\gamma \in \ClosN(\Lambda)
		\cap \ClosN(\Theta)$. 
		This yields that $\gamma$ is fixpoint-free. For any such $\gamma$ we apply a 
		$\RuCut^r$ rule with cut-formula $\gamma$, where $\Lambda = 
		\Lambda',\gamma$:\footnote{In addition we implicitly weakened all unimportant side-formulas in the left premiss and applied \RuExp rules to obtain the split sequent $\epsilon \proves \gamma^\epsilon \| \epsilon \proves \mybar{\gamma}^\epsilon$.}
		\[\begin{prooftree}
		\hypo{\gamma \| \mybar{\gamma}}
		\hypo{\theta^l \proves\ldiap \phi^\sigma, \lboxp \Sigma, \Lambda', \gamma^\tau 
		   \| \theta^r \proves \lboxp \Pi,\Theta, \gamma^{\epsilon}}
		\infer[left label = $\RuCut^r$]2[]{
		   \theta^{l} \proves\ldiap \phi^\sigma, \lboxp \Sigma, \Lambda', \gamma^\tau 
		   \| \theta^{r} \proves \lboxp \Pi,\Theta}
		 \end{prooftree}\]
	 	Applying the modal rule will now make $\ldiap[\conv{a}] \gamma $ land in the proper (right) component of the premiss.
	 	Likewise, applying a $\RuCut^l$ rule for every $\ldiap[\conv{a}]\delta^\epsilon \in \ldiap[\conv{a}] \Theta^\epsilon$, where $\ldiap[\conv{a}]\delta \in \ClosN(\Psi^l)\setminus \ClosN(\Psi^r)$ yields a split sequent, where we may apply $\Ru_{\ldiap}^l$ and satisfy conditions 1 -- 2.
	 	
	 	For trace atoms $\gamma \trace[k] \chi$  (and negated trace atoms $\gamma \ntrace[k] \chi$) occurring in $\Lambda^{\ldiap\phi}$ this is not a problem, as there are no trace atoms where $\gamma$ is fixpoint-free.
	 	
		\item Any discharge rule $\RuDischarge{\dd}$ is deleted. 
		\item \RuReset[] rules: By induction on $\pi$ we can show that in every annotated sequent no name $\dx$ occurs in both the left and the right component. This holds as only fresh names are introduced and in no rule do names cross the split. Thus \RuReset[] can always be translated to either $\RuReset[]^l$ or $\RuReset[]^r$.
		\item If the applied rule is \RuCut (or \RuTCut), add $\phi$ and $\mybar{\phi}$ (or $\phi \trace[k] \psi$ and $\phi \ntrace[k] \psi$) to the respective left components if $\phi$ is in $\ClosN$ of the left component of the conclusion and to the respective right components otherwise. 
		\item In the rule \RuJump it holds that $\phi^{\sigma}, \psi^{\tau}, \phi \trace[k] \psi$ are all either in $\ClosN(\Gamma)$ or in $\ClosN(\Delta)$, since all trace atoms are relevant. Similarly, for \RuTrans and \AxTraceNeg all explicitly written formulas in its conclusion belong to the same component of the sequent.
		\item All other rules have only one explicitly written formula in the conclusion and thus can easily be translated to a left or right rule.
	\end{itemize}
	Condition 3 can easily be satisfied by applying $\RuWeak^r$ if necessary. In order to satisfy conditions 4 and 5 we apply $\RuExp^l$ and $\RuExp^r$ rules whenever necessary.
	Thus we obtain a split \JStwo derivation satisfying the specified conditions.
	
	Next we want to fold the split \JStwo derivation $\pi$ into a split \JStwo proof. To do so the following claim is crucial:
	
	\claim{1}
	Let $\beta$ be an infinite branch in $\pi$. Then there is a name $\dx$ and $d = l,r$ such that 
	\begin{enumerate}
		\item $\dx$ occurs in the control of $\Sigma^d$ on cofinitely many sequents $\Sigma$ on $\beta$ and
		\item there are infinitely many applications of $\RuReset[]_\dx^d$ on $\beta$.
	\end{enumerate} 

	\claimproof{1}
	Let $\alpha$ be the corresponding infinite path of $\beta$ in $\rho$. Because $\alpha$ is successful there is a name $\dx$ such that 
	\begin{enumerate}
		\item $\dx$ occurs in the control of $\Sigma$ on cofinitely many sequents $\Sigma$ on $\alpha$ and
		\item there are infinitely many applications of $\RuReset[]_\dx$ on $\alpha$.
	\end{enumerate} 
	We can show inductively that $\dx$ either occurs in cofinitely many left, or in cofinitely many right components of $\beta$.
	If $\dx$ occurs in cofinitely many left components of $\beta$, then infinitely many $\RuReset[\dx]$ rules are translated to $\RuReset[]_{\dx}^l$ rules in $\pi$.  Analogously, if $\dx$ occurs in cofinitely many right components of $\beta$.
	\claimproofend
	
	Using Claim 1 we can fold $\rho$ into a split \JStwo proof using an analogous argument as in the proof of Lemma \ref{lem.InftyJSiffJS}.

	\smallskip
	Lastly we deal with the general case, where $\Gamma$ and $\Delta$ may share bound variables. Let $\Gamma'$ be an $\alpha$-equivalent sequent of $\Gamma$, where all bound variables in $\Gamma'$ and $\Delta$ are disjoint; for example replace every bound variable in $\Gamma$ by a fresh new variable not occurring in either $\Gamma$ or $\Delta$.
	By the above reasoning we obtain a split $\JStwo$ proof $\pi'$ of $\Gamma'\| \Delta$. In $\pi'$ we can translate back all newly introduced bound variables. This yields a split \JStwo proof $\pi$ of $\Gamma\| \Delta$.
\end{proof}

\begin{theorem}\label{thm.JStwoSplitSoundComplete}
	$\Gamma,\Delta$ is unsatisfiable iff there is a split \JStwo proof of $\Gamma \| \Delta$ for any pure sequents $\Gamma$ and $\Delta$.
\end{theorem}
\begin{proof}
	Lemma \ref{lem.JSinftyIffNW} together with Theorem \ref{thm.SoundnessNWtwo} and Theorem \ref{thm.CompleteenessNWtwo} yields that $\Gamma,\Delta$ is unsatisfiable iff there is a regular \JStwoInfty proof of $\Gamma,\Delta$.
	The soundness of split \JStwo proofs follows from Lemma \ref{lem.splitJStwoSoundness} and the completeness from Lemma \ref{lem.JStwoSplitCompleteness}.
\end{proof}

\section{Interpolation}\label{sec.interpolation}

In the previous section we saw that a sequent $\Gamma, \Delta$ is 
unsatisfiable iff there is a split \JStwo proof $\pi$ of $\Gamma \| \Delta$.
Given such a $\pi$ we will define an interpolant $I$ and construct split proofs $\pi^l$ of
$\Gamma \| I$ and $\pi^r$ of $\mybar{I} \| \Delta$.


\begin{theorem}[Craig interpolation]\label{thm.CraigInterpolation}
Let $\phi$ and $\psi$ be two $\muMLtwo$ formulas such that $\phi \modImpl \psi$.
Then there is an interpolant of $\phi$ and $\psi$.
\end{theorem}
\begin{proof}
Follows from Lemma \ref{lem.Interpolation}.
\end{proof}

\begin{corollary}[Beth definability]
	Let $p,q \in \Prop$ and let $\phi(p)$ be a $\muMLtwo$ formula. If\footnote{Here $\phi(q)$ is an abbreviation of $\phi(p)\subst{q}{p}$} $\phi(p),\phi(q) \modImpl p \liff q$, then there is a formula $\chi$ with $\Voc(\chi) \subseteq \Voc(\phi)\setminus\{p\}$ and $\phi(p) \modImpl p \liff \chi$.
\end{corollary}
\begin{proof}
	Apply Craig interpolation to $\phi(p), p \modImpl \phi(q) \impl q$.
\end{proof}

Let $\pi$ be a proof and $u$ be a companion node in $\pi$. A \emph{strongly connected subgraph} $S$ of $\pi$ is a strongly connected subgraph of $\calT_{\pi}^C$. 
The \emph{strongly connected subtree} $\scs(u)$ of $u$ is the maximal strongly connected subgraph $S$ of $\pi$, such that $u \in S$ and all nodes $v \in S$ are descendants of $u$ in $\calT_{\pi}$.

\begin{lemma}\label{lem.Interpolation}
Let $\pi$ be a split \JStwo proof of $\Gamma \| \Delta$. 
Then there is a formula $I$ such that $\Voc(I) \subseteq \Voc(\Gamma) \cap
\Voc(\Delta)$ and for which there are split \JStwo proofs $\pi^l$ of $\Gamma \| I$
and $\pi^r$ of $\mybar{I} \| \Delta$ .
\end{lemma}
\begin{proof} 
Let $D$ be the set of nodes in $\pi$ that are labelled by \RuDischarge{}. 
For a node $u \in \pi$ we define the set of \emph{active repeats} $A_u \isdef
\{c \in D \| u \in \scs(c) \}$.  
We define a priority function $\Omega_{\pi}: D \to \Nat^+$ satisfying 
$\Omega_{\pi}(c) < \Omega_{\pi}(d)$ if $c$ is a proper descendent of $d$ in $\calT_{\pi}$. 
Let $V_D \isdef \{x_c \| c \in D\}$ be a set of fresh new variables such 
that $x_c \in N_{k}$ with $k = \Omega_{\pi}(c)$ and define $V_u \isdef \{x_c \in V_D 
\| c \in A_u\}$. 
Our interpolant will be a formula with bound variables in $V_D$. 
		
For each node $u \in \pi$ labelled by $\theta_u \proves \Gamma_u \| \kappa_u \proves \Delta_u$ 
we define 
\begin{enumerate}
\item a formula $I_u$ with $\FV(I_u) \subseteq V_u$ and $\Voc(I_u) \subseteq \Voc(\Gamma_u) \cap \Voc(\Delta_u)$,
\item a derivation $\pi_u^l$ of $\theta_u \proves \Gamma_u \| \epsilon \proves I_u$ such that 
all open assumptions in $\pi_u^l$ are labelled by $\theta_c \proves \Gamma_c 
\| \epsilon \proves x_c$ for some $c \in A_u$ and  
\item a derivation $\pi_u ^r$ of $\epsilon \proves \mybar{I_u} \| \kappa_u \proves\Delta_u$ 
such that all open assumptions in $\pi_u^r$ are labelled by $\epsilon \proves
x_c \| \kappa_c \proves\Delta_c$ for some $c \in A_u$.
\end{enumerate}
We define $I_u,\pi_u^l, \pi_u^r$ by induction on the proof tree of $\pi$,
starting from the leaves. 
For the root $r$ of $\pi$ this will yield $I \isdef I_r$ such that $\Voc(I) \subseteq \Voc(\Gamma) \cap \Voc(\Delta)$ and proofs 
$\pi^l$ of $\Gamma \| I$ and $\pi^r$ of $\mybar{I} \| \Delta$. 
The construction is defined by a case distinction on the last applied rule. 
\begin{itemize}
	\item \emph{Axioms}: 
	If $u$ is labelled by an axiom of the form $\theta \proves \phi, \Gamma \|
	\mybar{\phi}, \Delta$, then $I_u \isdef \mybar{\phi}$ and dually $I_u \isdef \phi$ if $\phi$ and $\mybar{\phi}$ are
	swapped. 
	Otherwise an axiom is applied, where either $\phi,\mybar{\phi}$; $\bot$; $\phi 
	\trace[k] \psi, \phi \ntrace[k] \psi$ or $\phi \trace[2k] \phi$ is on the left 
	or the right side of the split. 
	If it is on the left, let $I_u \isdef \top$ and otherwise $I_u \isdef \bot$.
	It is straightforward to check the conditions 1--3.
	
\item \emph{Discharged leaves:} 
For every discharged leaf $u$ labelled by $\theta_u \proves \Gamma_u \| \kappa_u \proves \Delta_u$
with companion node $c$, let $I_u \isdef x_c$. 
Define $\pi_u^l$ to be the derivation consisting of one open assumption 
$\theta_u \proves \Gamma_u \| \epsilon \proves x_c$ and $\pi_u^r$ to be $\epsilon \proves 
x_c \| \kappa_u \proves \Delta_u$.
Clearly, the conditions 1--3 hold.
		
\item \emph{Companion nodes:} 
Let $u$ be labelled by $\RuDischarge{\dd}^r$ and let $v$ be its child. 
By induction hypothesis there is a formula $I_v$ and derivations $\pi_v^l$ and
$\pi_v^r$ satisfying conditions 1--3. 
We define $I_u \isdef \mu x_u. I_v$. 
In order to define $\pi_u^l$ we transform the derivation $\pi_v^l$ of 
$\theta_u \proves \Gamma_u \| \epsilon \proves I_u$. 
Let $O$ be the set of open assumptions in $\pi_v^l$ labelled by $\theta_u
\proves \Gamma_u \| \epsilon \proves x_u$ and let $P$ be the set of all other open assumptions. 

By uniformly substituting every occurrence of $x_u$ in $\pi_v^l$ by $\mu x_u. I_v$ we obtain a derivation $\rho_{v}$ of $\theta_u \proves \Gamma_u \| \epsilon \proves I_v\subst{\mu x_u. I_v}{x_u}$, where all open assumptions are either in $P$ or labelled by $\theta_u \proves \Gamma_u \| \epsilon \proves \mu x_u. I_v$. 
		
Let $\rho_v^{\dx}$ be obtained from $\rho_v$ by replacing every node $w$ in the strongly connected subtree $\scs(u)$ of $u$ labelled by $\theta_w \proves \Gamma_w \| \sigma(w) \proves I_w^{\sigma(w)}$ with $w'$, where $w'$ is labelled by the sequent $\theta_w \proves \Gamma_w \| \dx \cdot \sigma(w) \proves I_w^{\dx \sigma(w)}$. If a node $w$ is not in the strongly connected subtree of $u$, but its parent is, then add an \RuExp rule to remove the name $\dx$. This results in a well-formed derivation because $\dx$ is a higher-ranking name than all names in $\sigma(w)$ for all $w$.
		
We define the following derivation $\pi_u^l$, where all open assumptions are 
in $P$ and all assumptions from $O$ are discharged as follows.
\[
\begin{prooftree}
\hypo{\discharge{\theta_u \proves \Gamma_u \| \dx \proves I_v\subst{\mu x_u. I_v}{x_u}^{\dx}}{\de}}				
\infer[left label= $\RuReset[]_\dx^r$]1{\theta_u \proves \Gamma_u \| \dx \dy \proves I_v\subst{\mu x_u. I_v}{x_u}^{\dx\dy}}
\infer[left label=$\RuMu^r$]1[]{\theta_u \proves \Gamma_u \| \dx \proves \mu x_u. I_v^{\dx}}
\infer[no rule]1{\vdots}
\infer[no rule]1{\rho_v^\dx}
\infer[no rule]1{\vdots}
\infer[no rule]1[]{\theta_u \proves \Gamma_u \| \dx \proves I_v\subst{\mu x_u. I_v}{x_u}^{\dx}}
\infer[left label=$\RuDischarge{\de}^r$]1[]{\theta_u \proves \Gamma_u \| \dx \proves I_v\subst{\mu x_u. I_v}{x_u}^{\dx}}
\infer[left label=$\RuMu^r$]1[]{\theta_u \proves \Gamma_u \| \epsilon \proves \mu x_u. I_v}
\end{prooftree}
\]
	For the definition of $\pi_u^r$ we let $\rho_v^r$ be obtained from $\pi_v^r$ by uniformly substituting every occurrence of $x_u$ in $\pi_v^r$ by $\nu x_u. \mybar{I_v}$. We let $\pi_u^r$ be the following derivation
	\begin{center} 
		\begin{prooftree}
			\hypo{\discharge{\epsilon \proves \nu x_u. \mybar{I_v} \| \kappa_u \proves \Delta_u}{\df}}
			\infer[no rule]1{\vdots}
			\infer[no rule]1{\rho_v^r}
			\infer[no rule]1{\vdots}
			\infer[no rule]1{\epsilon \proves \mybar{I_v}\subst{\nu x_u. \mybar{I_v}}{x_u} \| \kappa_u \proves \Delta_u}
			\infer[left label=$\RuNu^l$]1[]{\epsilon \proves \nu x_u. \mybar{I_v} \| \kappa_u \proves\Delta_u}
			\infer[left label=$\RuDischarge{\df}^r$]1[]{\epsilon \proves \nu x_u. \mybar{I_v} \| \kappa_u \proves \Delta_u}
		\end{prooftree}
	\end{center} 
It holds that $\FV(I_u) = \FV(I_v)\setminus\{x_u\} \subseteq V_u$ and $\Voc(I_u) = \Voc(I_v)$, thus the conditions 1--3 are satisfied.	
The case where a $\RuDischarge{}^l$ rule is applied is dual with $I_u \isdef \nu x_u. I_v$.
	
	\item \emph{Modal rules:} Let $u$ be labelled by $\RuDiaR$ and let $v$ be its child. 
If the left component of $v$ is empty, define $I_u \isdef \bot$, then $\pi_u^l$ is an instance of \AxBot and $\pi_u^r$ is obtained from $\pi_v$ by applications of \RuDia and \RuWeak. The conditions 1--3 are clearly satisfied.
	
	Otherwise both components of the premiss of \RuDia in $\pi_u$ are non-empty. Then it follows that the action $a$ belongs to the vocabulary of both $\Gamma_u$ and $\Delta_u$. To see that for the left component, let $\Gamma_v = \Sigma, \ldiap[\conv{a}] \Pi_u^{\epsilon}$ be non-empty. If $\Sigma$ is non-empty, then clearly $a \in \Voc(\Gamma_u)$. Otherwise there is  $\ldiap[\conv{a}] \gamma^{\epsilon} \in \ldiap[\conv{a}]\Pi_u^{\epsilon}$, yet this is only the case if $\ldiap[\conv{a}] \gamma \in \ClosN(\Gamma_u)$, which implies $a \in \Voc(\Gamma_u)$ indeed. 
	
	We define $I_u \isdef \ldiap I_v$, The proofs $\pi_u^d$ are obtained from $\pi_v^d$ by applying a \RuDiaR rule for $d = l,r$. It holds that $\Voc(I_u) = \Voc(I_v) \cup \{a\} \subseteq \Voc(\Gamma_u) \cap \Voc(\Delta_u)$ and therefore the conditions 1--3 are satisfied. The case of a $\RuDiaL$ rule is dual with $I_u \isdef \top$ or $I_u \isdef \lboxp I_v$.
	
	\item \emph{Unary rules:} If $u$ is the conclusion of a unary rule different than $\RuDischarge{}$ and \RuDia with premiss $v$, define $I_u \isdef I_v$. The proofs $\pi_u^l$ and $\pi_u^r$ are defined straightforwardly. 
	\item \emph{Binary rules:} 
	If $u$ is the conclusion of a binary rule $\Ru$ with premisses $v$ and $w$, than $I_u \isdef I_v \land I_u$ or $I_u \isdef I_v \lor I_w$, depending on whether $\Ru$ is a left or right rule. The proofs $\pi_u^l$ and $\pi_u^r$ are defined straightforwardly. For example, if $\Ru = \RuCut^l$, then $\pi_u^l$ is the following proof. Note that in applications of \RuCut, \RuAnd and \RuOr the control remains the same, we therefore omit the control in the following proof due to space constraints.  Additionally, we apply $\RuWeak^r$ rules  implicitly on both branches.
	\[ 
	\begin{prooftree}
		\hypo{\pi_v^l}
		\infer[no rule]1{ \phi^\epsilon, \Gamma_u \|  I_v}
		\infer[left label=$\RuAnd^r$]1[]{ \phi^\epsilon, \Gamma_u \|  I_v \land I_w}
		\hypo{\pi_w^l}
		\infer[no rule]1{ \mybar{\phi}^\epsilon, \Gamma_u \|  I_w}
		\infer[left label=$\RuAnd^r$]1[]{ \mybar{\phi}^\epsilon, \Gamma_u \| I_v \land I_w}
		\infer[left label=$\RuCut^l$]2[]{ \Gamma_u \|  I_v \land I_w}
	\end{prooftree}
	\] 
	
	The proof $\pi_u^r$ is defined as follows.
	\[ 
	\begin{prooftree}
		\hypo{\pi_v^r}
		\infer[no rule]1{\theta_v \proves \mybar{I_v} \| \kappa_v \proves \Delta_u}
		\hypo{\pi_w^r}
		\infer[no rule]1{\theta_w \proves \mybar{I_w} \| \kappa_w \proves \Delta_u}
		\infer[left label=$\RuOr^l$]2{\theta_u \proves \mybar{I_v} \lor \mybar{I_w} \| \kappa_u \proves \Delta_u}
	\end{prooftree}
	\] 
	As every application of \RuCut is analytic it holds that $\FV(\phi) \subseteq \FV(\Gamma_u)$. Therefore $\Voc(I_v \land I_w) \subseteq \Voc(\Gamma_u,\phi) \cap \Voc(\Delta_u) = \Voc(\Gamma_u) \cap \Voc(\Delta_u)$, hence conditions 1--3 are satisfied.		

	\end{itemize}
\end{proof}

\section{Conclusions and Questions}
\label{sec.conclusion}

The main contribution of this paper is the result that the two-way $\mu$-calculus
has the Craig interpolation property.

Below we mention some questions for further research.
\begin{enumerate}
\item Do natural \emph{fragments} of $\muMLtwo$ have interpolation?
Based on Rooduijn \& Venema's focus system an affirmative answer seems within 
reach for the alternation-free fragment.
Building on recent developments on the interpolation problem for propositional 
dynamic logic~\cite{borz:prop25}, in collaboration with Valentina Trucco Dalmas
we recently established the Craig Interpolation Property for Converse 
PDL~\cite{kloi-inte25}.
\item
Since the two-way $\mu$-calculus does not have the finite model property, an 
interesting question is whether it has interpolation \emph{in the finite}, that
is, for the finite-model-theory version $\modImpl_{f}$ of $\modImpl$ given by
$\phi \modImpl_{f} \psi$ if $\bbS,s \sat \phi$ implies $\bbS,s \sat \psi$
for every \emph{finite} pointed model $\bbS,s$.
\item
In the introduction we mentioned work by Benedikt and collaborators on 
\emph{guarded fixpoint logics}~\cite{bene:inte15,bene:defi19}.
It would be interesting to compare these results to ours, and to see whether 
our approach could lead to proof systems for their logics, or 
whether their model-theoretic approach would also work for the two-way 
$\mu$-calculus.
\item
A similar question applies to the work of French~\cite{fren:bisi06,fren:idem07}. 
Given the connection between uniform interpolation and bisimulation 
quantifiers~\cite{dago:logi00}, French's results might even lead to an (indirect)
proof that the two-way $\mu$-calculus has the Uniform Interpolation Property.
\item
The uniform interpolation property of $\muML$~\cite{dago:logi00} is related to the fact that every 
$\muML$-formula has a so-called \emph{disjunctive normal form}~\cite{Janin1995}
using the cover modality $\nabla$.
Can we prove a similar result for the two-way $\mu$-calculus?
\item
To the best of our knowledge many other natural questions about the two-way
$\mu$-calculus are open as well, for instance: is $\muMLtwo$ the fragment of
monadic second-order logic that is invariant under two-way bisimulations?
\end{enumerate}

\bibliographystyle{IEEEtran}
\bibliography{TeX/bib.TwoWay,TeX/tw-extra}

\appendix

\section{Appendix}\label{sec.appendix}

\subsection{Parity games}\label{app.parityGames}
In this appendix we briefly define infinite two-player games, for more details we refer to \cite{Graedel2002}.
We fix two players that we shall refer to as $\eloi$ (Eloise, female) and $\abel$
(Abelard, male) and use $\Pi$ as a variable ranging over the set $\{ \eloi, \abel\}$.

A \emph{two-player game} is a quadruple $\bbG = (V,E,\Own,\WC)$ where $(V,E)$ is a graph, $\Own$ is a map $V \to \{ \eloi, \abel \}$, and $\WC$ is a set of infinite paths in $(V,E)$.
	An \emph{initialised game} is a pair consisting of a game $\bbG$ and an element
	$v$ of $V$, usually denoted as $\bbG@v$.

We will refer to $(V,E)$ as the \emph{board} of the game. 
Elements of $V$ will be called \emph{positions}, and $\Own(v)$ is the 
\emph{owner} of $v$.
 Given a position $v$ for player $\Pi$, the set $E[v]$ 
denotes the set of \emph{moves} that are \emph{admissible
	for} $\Pi$ at $v$. We denote $V_{\Pi} \isdef \Own^{-1}(\Pi)$.
The set $\WC$ is called the \emph{winning condition} of the game.

A \emph{match} of the game $\bbG = (V,E,\Own,\WC)$
is a path $\pi$ through the graph $(V,E)$.
Such a match $\pi$ is \emph{full} if it is maximal as a path, that is, either
finite with $E[\last(\pi)] = \nada$, or infinite.
If a position has no $E$-successors, the owner of that positions \emph{gets stuck} and 
loses the match.
Infinite matches are won by $\eloi$ if the match, as an $E$-path, belongs to the set $\WC$ and won by $\abel$ otherwise.


Let $\PM_{\Pi}$ denote the collection of partial matches $\pi$ ending in a position\footnote{For a finite sequence $s = v_0...v_n$ we define $\first(s) \isdef v_0$ and $\last(s) \isdef v_n$.} $\last(\pi) \in V_{\Pi}$.
A \emph{strategy} for a player $\Pi$ is a partial function $f: \PM_{\Pi} \to V$ such that $f(\pi)
\in E[\last(\pi)]$ if $E[\last(\pi)] \neq \nada$.
A match $\pi = (v_{i})_{i<\kappa}$ is \emph{guided} by a $\Pi$-strategy $f$, in short $f$-guided, if $f(v_{0}v_{1}\cdots v_{n-1}) = v_{n}$ for all $n<\kappa$ 
such that $v_{0}\cdots v_{n-1}\in \PM_{\Pi}$.
%
A $\Pi$-strategy $f$ is \emph{winning for $\Pi$ from $v$} if $\Pi$ wins all $f$-guided full
matches starting at $v$.
The game $\bbG$ is \emph{determined} if every position is winning for either 
$\eloi$ or $\abel$.

A strategy is \emph{positional} if it only depends on the last position of a 
partial match, namely, if $f(\pi) = f(\pi')$  whenever $\last(\pi) = \last(\pi')$;
such a strategy can and will be presented as a map $f: V_{\Pi} \to V$.

A \emph{parity game} is a board game $\bbG = (V,E,\Own,\WC_\Omega)$ in which the
winning condition $\WC_\Omega$ is given by a priority map $\Omega: V \to \Nat$ as follows: $\pi \in \WC_\Omega$ iff $\max\{\Omega(v) \| v \text{ occurs infinitely often in } \pi \}$ is even.
Such a parity game is usually denoted as $\bbG = (V,E,\Own,\Om)$.
The following theorem is independently due to Emerson \& Jutla~\cite{Emerson99}
and Mostowski~\cite{Mostowski1991}.

\begin{theorem}[Positional Determinacy]
	\label{thm.gamesPosDeterm}
	Let $\bbG = (G,E,\Own,\Om)$ be a parity game.
	Then $\bbG$ is determined, and both players have positional winning strategies.
\end{theorem}

\subsection{Soundness of \NWtwo}\label{app.soundness}
The following lemma deals with the local soundness of our rules and can be 
proven straightforwardly:
\begin{lemma}\label{lem.soundLocal}
	Let 
	\[ 
	\begin{prooftree}
		\hypo{\Delta_1}
		\hypo{\cdots}
		\hypo{\Delta_m}
		\infer[left label =\Ru]3{\Gamma}
	\end{prooftree}
	\] 
	be a rule instance of Figure \ref{fig.NWtwo}. 
	If $\Gamma$ is satisfiable, then there is an $i = 1,...,m$ such that $\Delta_i$ 
	is satisfiable. 
	
	In particular, if $\Ru \neq \RuDia$, and given a pointed model $\bbS,s$ and positional strategy
	$f$ for $\eloi$ in $\calE(\Land\Phi, \bbS)$ such that $\bbS,s \sat_f \Gamma$, then $\bbS,s 
	\sat_f \Delta_i$. 
	If $\Ru = \RuOr$ and $\phi_0 \lor \phi_1$ is the principal formula, then 
	$\bbS,s \sat_f \phi_i, \phi_0 \lor \phi_1 \trace[1] \phi_i, \Gamma$, where 
	$f(\phi_0 \lor \phi_1,s) = (\phi_i,s)$.
	
	If $\Ru = \RuDia$ with principal formula $\ldiap \phi$ and given a pointed model 
	$\bbS,s$ and positional strategy $f$ for $\eloi$ in $\calE(\Land\Phi, \bbS)$ such that $\bbS,s \sat_f 
	\ldiap \phi, \lboxp \Sigma, \Gamma$, then $\bbS,t \sat_f \phi,\Sigma, 
	\ldiap[\conv{a}]\Gamma, \Gamma^{\ldiap\phi}$, where $f(\ldiap \phi,s) = 
	(\phi,t)$. 
\end{lemma}

\setcounter{theorem}{\getrefnumber{thm.SoundnessNWtwo}}
\addtocounter{theorem}{-1}
\begin{theorem}
	If $\NWtwo \proves \Gamma$, then $\Gamma$ is unsatisfiable.
\end{theorem}
\begin{proof}
	By contraposition we show that, if $\Gamma$ is satisfiable, then Builder has a 
	winning strategy in $\calG := \calG(\Phi)@\Gamma$. 
	So assume that there is a pointed model $\bbS,s$ and a positional strategy $f$ for $\eloi$ 
	in the game $\calE := \calE(\Land\Phi, \bbS)$ such that $\bbS,s \sat_f \Gamma$.
	We will  construct a winning strategy $\overline{f}$ for Builder in $\calG$ and
	a function $s_f: \PM(\Phi) \to \bbS$, mapping partial $\calG$-matches to states 
	of $\bbS$, such that $\bbS,s_f(\calM) \sat_f \last(\calM)$ for every 
	$\overline{f}$-guided $\calM \in \PM_P(\Phi)$.
	
	The functions $\overline{f}$ and $s_f$ can be defined inductively by a case distinction based on the rule instance. For the base case $|\calM| = 1$ it holds that $\calM = \Gamma$. We define $s_f(\calM) \isdef s$ and do not have to define $\overline{f}$ as this is a position owned by Prover. Otherwise we follow the specifications of the rule instance. If the rule is $\RuDia$, define $s_f$ as given by $f$ and let $\overline{f}$ choose the only premiss. For any other rule $s_f$ remains the same and we invoke Lemma \ref{lem.soundLocal} for the definition of $\overline{f}$.
	
	We need to show that $\overline{f}$ is a winning strategy for Builder in $\calG$.
	Because of Lemma \ref{lem.soundLocal} we know that all finite matches are won by
	Builder. 
	Thus, assume by contradiction that Prover wins an infinite 
	$\overline{f}$-guided $\calG$-match $\calM$. 
	Then there is a $\mu$-trail $\tau= \tau_0\tau_1\cdots$ on $\calM$.
	Note that $\tau$ is not necessarily a trail starting from the root; it might
	also be starting from a cut formula. We will use $\tau$ to obtain an infinite
	$f$-guided $\calE$-match $\calN$ that is won by $\abel$.
	
	Let $\tau_i = (\phi_i,\psi_i,k_i)$ for $i \geq 0$, remember that $\phi_{i+1} = \psi_i$ for all $i$. Let $\calM_i$ be the initial partial match of $\calM$, such that $\tau_i$ ``is a trace starting at $\last(\calM_i)$''. We will define $f$-guided partial $\calE$-matches $\calN_i$ starting at $(\phi_i,s_f(\calM_i))$ and ending at $(\phi_{i+1},s_f(\calM_{i+1}))$ for every $i \geq 0$, such that $k_i = \max\{\Omega(\phi) \| (\phi,s) \text{ is a position in } \calN_i \text{ for some } s\}$.
	
	If $\tau_i$ is an upward trace in $\sfT_{u,v}$ for $u \neq v$, then we can
	define $\calN_i$ straightforwardly. 
	Otherwise $\tau_i$ is a detour trace in $\sfT_{u,u}$ for some $u$. 
	Then $\phi_i \trace[k_i] \phi_{i+1} \in \last(\calM_i)$. As $\bbS, s_f(\calM_i) 
	\sat_f \last(\calM_i)$ it holds that $\bbS, s_f(\calM_i) \sat_f \phi_i 
	\trace[k_i] \phi_{i+1}$. 
	This exactly means that there is an $f$-guided match $\calN_i$ starting at 
	$(\phi_i,s_f(\calM_i))$ and ending at $(\phi_{i+1},s_f(\calM_{i}))$ as needed.
	
	Glueing together the matches $\calN_i$ we obtain an infinite $f$-guided $\calE$-match
	$\calN= \calN_0\calN_1\cdots$, such that $\max\{\Omega(\phi) \| \phi 
	\text{ occurs infinitely often in } \calN\} = \max\{k \| k \text{ appears 
		infinitely often on } \tau\}$. Thus we conclude that $\calN$ is won by $\abel$ and therefore $\bbS,s_f(\calM_0) \not\sat_f \last(\calM_0)$,
	which is the desired contradiction.
\end{proof}

\subsection{Completeness of \NWtwo}\label{app.completeness}

\setcounter{theorem}{\getrefnumber{lem.soundLocal}}
\begin{lemma}[Saturation]\label{lem.NWcomplSaturation}
	For every state $\rho$ in $\bbS^f$ the set $\sfS(\rho)$ is \emph{saturated}, meaning that the following conditions are satisfied:
	\begin{enumerate}
		\item For all $\phi \in \ClosN(\sfS(\rho))$ it holds that $\phi \in 
		\sfS(\rho)$ iff $\mybar{\phi} \notin \sfS(\rho)$.
		\item Never $\bot \in \sfS(\rho)$.
		\item For all $\phi \in \ClosN(\sfS(\rho))$ and relevant trace atoms 
		$\phi \trace[k] \psi$ it holds that $\phi \trace[k] \psi \in \sfS(\rho)$ or
		$\phi \ntrace[k] \psi \in \sfS(\rho)$.
		\item For no relevant trace atom
		$\phi \trace[k] \psi$ it holds that $\phi \trace[k] \psi \in \sfS(\rho)$ and
		$\phi \ntrace[k] \psi \in \sfS(\rho)$.
		\item For no $\phi, k $ it holds that $\phi \trace[2k] \phi \in \sfS(\rho)$.
		\item If $\phi_0 \land \phi_1 \in \sfS(\rho)$, then for both $i=0,1$ it holds
		that $\phi_i \in \sfS(\rho)$ and $\phi_0\land \phi_1 \trace[1] \phi_i \in 
		\sfS(\rho)$.
		\item If $\phi_0 \lor \phi_1 \in \sfS(\rho)$, then for some $i=0,1$ it holds 
		that $\phi_i  \in \sfS(\rho)$ and $\phi_0\lor \phi_1 \trace[1] \phi_i \in 
		\sfS(\rho)$.
		\item If $\eta x. \phi \in \sfS(\rho)$, then $\phi[\eta x . \phi / x], \eta x.
		\phi \trace[\Omega(\eta x. \phi)] \phi[\eta x . \phi / x] \in \sfS(\rho)$.
		\item If $\phi \trace[k] \psi, \psi \trace[l] \chi \in \sfS(\rho)$, then
		$\phi \trace[\max\{k,l\}] \chi \in \sfS(\rho)$.
	\end{enumerate}
\end{lemma}
\begin{proof}
	Follows from our restriction on the strategy of Prover. 
\end{proof}

\begin{lemma}\label{lem.thruthLemmaFormulas}
	Let $\rho_0$ be a state of $\bbS^f$ containing the root $\Gamma$ of $\calT$ and let $\psi_0 \in \Gamma$.
	Let $\calM$ be an $\underline{f}$-guided $\calE$-match with starting position $(\psi_0,\rho_0)$. Then for every position $(\psi, \rho)$ in $\calM$ it holds that $\psi \in \sfS(\rho)$.
\end{lemma}
\begin{proof}
	We write $(\psi_n,\rho_n)$ for the $n$-th position of $\calM$ and prove the claim by strong induction on $n$. The base case is clear. For the induction step let $\psi_n \in \sfS (\rho_n)$, we have to show that $\psi_{n+1} \in \rho_{n+1}$. We proceed by a case distinction based on the shape of $\psi$. If $\psi$ is not a modal formula, then $\rho_{n+1} = \rho_n$ and the claim follows from Lemma \ref{lem.NWcomplSaturation} and the definition of $\underline{f}$.
	
	Now assume that $\psi_n = \lboxp \chi$, then $\psi_{n+1} = \chi$. In this case $\rho_n R_a^f \rho_{n+1}$, so either $\rho_n \to[a] \rho_{n+1}$  or  $\rho_{n+1} \to[\conv{a}]\rho_{n}$. If $\rho_n \to[a] \rho_{n+1}$, then $\lboxp \chi$ is in the conclusion of \RuDia, hence $\chi$ is in its premiss and thus $\chi \in \rho_{n+1}$. 
	
	Because $\bbS^f$ is a forest and $\rho_0...\rho_{n+1}$ forms a path in $\bbS^f$ starting at one of the roots, where the last step of the path is downwards, there has to be an $i \in \{ 0,...,n-1 \}$ with $\rho_i = \rho_{n+1}$. 
	As $\calM$ is a match with positions $(\psi_i,\rho_i)$ and $(\lboxp \chi, \rho_n)$ for some $\psi_i$, it holds that $\lboxp \chi \in \Clos(\psi_i)$. As by induction hypothesis $\psi_i \in \sfS(\rho_i) = \sfS(\rho_{n+1})$, this yields $\lboxp \chi \in \Clos(\sfS(\rho_{n+1}))$ and also $\chi \in \Clos(\sfS(\rho_{n+1}))$.
	
	Towards a contradiction assume that $\chi \notin \sfS(\rho_{n+1})$. Because $\chi \in \ClosN(\sfS(\rho_{n+1}))$ it holds that $\mybar{\chi} \in \sfS(\rho_{n+1})$ by Lemma \ref{lem.NWcomplSaturation}. If $\mybar{\chi}$ is in the conclusion of $\RuDia[\conv{a}]$, then $\ldiap \mybar{\chi}$ is in its premiss as $\ldiap \mybar{\chi} \in \ClosN(\sfS(\rho_{n+1}))$, therefore $\ldiap \mybar{\chi} \in \rho_n$. Again by Lemma \ref{lem.NWcomplSaturation} we conclude that $\lboxp \chi \notin \rho_n$, which is a contradiction.
	
	Finally the case where $\psi_n = \ldiap \chi$ is similar to the first direction
	of the previous case. 
\end{proof}

\begin{lemma}\label{lem.truthLemmaTraces}
	Let $\rho \in S^f$, $\phi \in \sfS(\rho)$ and $\phi \trace[k] \psi$ be a relevant trace atom. If $\bbS^f,\rho \Vdash_{\underline{f}} \phi \trace[k] \psi$, then $\phi \trace[k] \psi \in \sfS(\rho)$.
\end{lemma}

\begin{proof}
	Let $\calN$ be an $\underline{f}$-guided $\calE$-match witnessing $\bbS^f,\rho \Vdash_{\underline{f}} \phi \trace[k] \psi$, meaning that $\calN$ is of the form
	\[
	(\phi,\rho) = (\phi_0,\rho_0) \cdots (\phi_n,\rho_n) = (\psi,\rho), \qquad n > 0
	\] 
	such that $k = \max\{\Omega(\phi_i) \| i =0,\ldots,n-1\}$.
	
	We prove the lemma by an induction on the number of distinct states occurring in $\calN$. For the base case, where the only state occurring in $\calN$ is $\rho$, we proceed with an inner induction on the length of $\calN$. The claim then follows straightforwardly from the definition of $\underline{f}$ and the fact that $\sfS(\rho)$ is saturated (Lemma \ref{lem.NWcomplSaturation}). 
	
	For the induction step let $\calN = \calA_1 \calB_1 \cdots \calB_{m-1} \calA_m$, such that for any position $(\chi,\tau)$ in $\calN$ we have: If $\tau = \rho$, then $(\chi,\tau) \in \calA_i$ for some $i = 1,...,m$; and if $\tau \neq \rho$, then $(\chi,\tau) \in \calB_i$ for some $i = 1,...,m-1$. Because $\bbS^f$ is a forest it follows that there is $\tau_i$ such that $\first(\calB_i) = (\beta_i,\tau_i)$ and $\last(\calB_i) = (\delta_i,\tau_i)$ for $i = 1,...,m-1$. We will fix the following notation for each $i$: 
	\begin{align*}
		\first(\calA_i) &= (\alpha_i,\rho)  \qquad &\first(\calB_i) &= (\beta_i,\tau_i),\\
		\last(\calA_i) &= (\gamma_i,\rho)  \qquad &\last(\calB_i) &= (\delta_i,\tau_i),
	\end{align*}
	By the base case of the induction it holds that $\alpha_i \trace[k_i] \gamma_i \in \sfS(\rho)$ for some $k_i \leq k$. For readability we will omit the subscripts in the trace atoms in the rest of the proof, for instance we will write $\alpha_i \trace \gamma_i$ instead of $\alpha_i \trace[k_i] \gamma_i$. As we closely follow the match $\calN$ we will end up with a trace atom of the form $\phi \trace[k] \psi$, where $k$ is as required. In the matches $\calB_i$ less states occur than in $\calN$, therefore we may apply the induction hypothesis to obtain $\beta_i \trace \delta_i \in \sfS(\tau_i)$ for $i = 1,...,m-1$. We aim to show that $\gamma_i \trace \alpha_{i+1} \in \sfS(\rho)$ for $i = 1,...,m-1$.
	
	Since the match $\calN$ transitions from $\rho$ to $\tau_i$ it holds that $\rho R_a^f \tau_i$ for some action $a$. First consider the case that $\rho \to[a] \tau_i$, then $\gamma_i$ must be of the form $\ldiap[a] \beta_i$ or of the form $\lboxp[a] \beta_i$. Looking at the transition from $\tau_i$ to $\rho$, because by the definition of $\underline{f}$ Eloise only moves upwards in $\bbS^f$, it must hold that $\delta_i = \lboxp[\conv{a}]\alpha_{i+1}$. We then obtain:
	\begin{align*}
		\beta_i &\trace \lboxp[\conv{a}]\alpha_{i+1} \in \sfS(\tau_i) \quad &&(\text{Induction hypothesis})\\
		\Rightarrow \beta_i &\ntrace \lboxp[\conv{a}]\alpha_{i+1} \notin \sfS(\tau_i) ~ &&(\text{Saturation})\\
		\Rightarrow \gamma_i &\ntrace \alpha_{i+1} \notin \sfS(\rho) ~ &&(\text{Definition of }\RuDia)\\
		\Rightarrow \gamma_i &\trace \alpha_{i+1} \in \sfS(\rho) ~ &&(\text{Saturation})
	\end{align*}
	Note that in the second implication we use the fact that $\lboxp[\conv{a}]\alpha_{i+1} \in \ClosN(\sfS(\rho))$. This follows as $\phi \in \sfS(\rho)$ by assumption and $\lboxp[\conv{a}]\alpha_{i+1} \in \Clos(\phi)$.
	In the third implication we rely on $\gamma_i \in \Clos(\phi) \subseteq \ClosN(\sfS(\rho))$.
	
	Now consider the case that $\tau_i \to[\conv{a}] \rho$. Then $\gamma_i = \lboxp[a]\beta_i$ for some action $a$ and $\delta_i$ is of the form $\ldiap[\conv{a}]\alpha_{i+1}$ or $\lboxp[\conv{a}]\alpha_{i+1}$. Here we find:
	\begin{align*}
		\beta_i &\trace \delta_i \in \sfS(\tau_i) \quad &&(\text{Induction hypothesis})\\
		\Rightarrow \lboxp[a]\beta_i &\trace \alpha_{i+1} \in \sfS(\rho) ~ &&(\text{Definition of }\RuDia)
	\end{align*}
	For this implication to hold it is required that $\lboxp[a]\beta_i \in \ClosN(\sfS(\tau_i))$. As $\lboxp[a]\beta_i \in \Clos(\sfS(\rho))$ and all \NWtwo rules are analytic, this implies $\lboxp[a]\beta_i \in \ClosN(\sfS(\tau_i))$ indeed.
	
	In both cases $\gamma_i \trace \alpha_{i+1} \in \sfS(\rho)$ and, as also $\alpha_i \trace \gamma_i \in \sfS(\rho)$ for $i = 1,...,m$, we can combine these statements using saturation to obtain $\alpha_1 \trace \gamma_m \in \sfS(\rho)$. Because we closely followed the match $\calN$ this yields $\phi \trace[k] \psi \in \sfS(\rho)$ as required.
\end{proof}

\begin{proposition}\label{prop.Completeness}
	Let $\rho_0$ be a state of $\bbS^f$ containing the root $\Gamma$ of $\calT$ and let $\psi_0 \in \Gamma$. Then the strategy $\underline{f}$ is winning for $\eloi$ in $\calE@(\psi_0,\rho_0)$.
\end{proposition}
\begin{proof}
	Let $\calM$ be an arbitrary $\underline{f}$-guided $\calE@(\psi_0,\rho_0)$-match.
	If $\calM$ is a finite match, then it is straightforward to check that it is won
	by $\eloi$.
	
	Suppose that $\calM= (\psi_n,\rho_n)_{n \in \omega}$ is infinite, and to arrive
	at a contradiction assume that $\abel$ wins $\calM$.
	By positional determinacy we may assume that his strategy is positional. We make a case distinction.
		
	First assume that there is a state $\rho$ that is visited infinitely often. 
	Then there must be a segment $\calN$ of $\calM$ such that $\first(\calN) =
	\last(\calN) = (\psi,\rho)$ for some formula $\psi$. 
	As the match is positional this means that $\calM = \calK \calN^*$ for some 
	initial segment $\calK$ of $\calM$, meaning that only finitely many states are 
	visited.
	By our assumption the match $\calM$ is winning for $\abel$, thus the most 
	important fixpoint formula occurring infinitely often is of the form $\mu x.
	\psi$. 
	Let $k = \Omega(\mu x. \psi)$. Because only finitely many states are visited, 
	there has to be a position $(\tau, \mu x. \psi)$ occurring infinitely often in 
	$\calM$ and thus $\bbS^f, \tau \sat_{\underline{f}} \mu x. \psi \trace[k] 
	\mu x. \psi$. 
	Then Lemma \ref{lem.thruthLemmaFormulas} yields $\mu x. \psi \in \sfS(\tau)$ and Lemma \ref{lem.truthLemmaTraces} gives $\mu x. \psi \trace[k] 
	\mu x. \psi \in \sfS(\tau)$. But in this case \AxTraceLoop would be applicable, 
	contradicting the assumption that $\calM$ is infinite.
	
	Now consider the case that $\calM= (\psi_n,\rho_n)_{n \in \omega}$ visits each state at most finitely often. Then there are sequences of indices $(\alpha(n))_{n \in \omega}, (\beta(n))_{n \in \omega} \in \omega^\omega$, such that 
	\begin{itemize}
		\item $\alpha(n) \leq \beta(n)$ and $\beta(n)+1 = \alpha(n+1)$ for all $n \in \omega$,
		\item $\rho_{\alpha(n)} = \rho_{\beta(n)}$ for all $n \in \omega$ and
		\item $\psi_{\beta(n)}$ is modal for every $n \in \omega$ and there is an action $a_n$ such that $\rho_{\beta(n)} \to[a_n] \rho_{\alpha(n+1)}$.
	\end{itemize}
	These indices can be defined by induction rather straightforwardly.
	
	Again assume that $\calM$ is winning for $\abel$, then there is $N \in \omega$ 
	such that for some even $k$ it holds that $\Omega(\psi_n)\leq k$ for all $n \geq N$ 
	and $\Omega(\psi_n) = k$ for infinitely many $n \geq N$. 
	Therefore it holds that $\bbS^f, \rho_{\alpha(n)} \sat_{\underline{f}} 
	\psi_{\alpha(n)} \trace[k(n)] \psi_{\beta(n)}$, where $k(n) \leq k$ for all 
	$n \geq N$ and $k(n) = k$ for infinitely many $n> N$. 
	Lemma \ref{lem.truthLemmaTraces} together with Lemma \ref{lem.thruthLemmaFormulas} yields that $\psi_{\alpha(n)} \trace[k(n)]
	\psi_{\beta(n)} \in \sfS(\rho_{\alpha(n)})$ and because Prover only applies cumulative rules in $\calT$ this implies $\psi_{\alpha(n)} \trace[k(n)]
	\psi_{\beta(n)} \in \last(\rho_{\alpha(n)})$.
	
	Clearly there is a trail $\tau_n$ from $\psi_{\beta(n)}$ at $\last(\rho_{\alpha(n)})$ to $\psi_{\alpha(n+1)}$ at $\first(\rho_{\alpha(n+1)})$ of weight $1$; there is only one modal rule applied. 
	Again because Prover only applies cumulative rules in $\calT$ there are trails $\tau_n'$ of weight 1 from $\psi_{\alpha(n)}$ at $\first(\rho_{\alpha(n)})$ to $\psi_{\alpha(n)}$ at $\last(\rho_{\alpha(n)})$.
	
	Thus we obtain the weighted trail
	\[\tau = (\psi_{\alpha(0)}, \psi_{\beta(0)},k(0))\cdot \tau_0 \cdot \tau_1' \cdot (\psi_{\alpha(1)}, \psi_{\beta(1)},k(1))\cdot \tau_1\cdots\]
	where  $\max\{l \| l\text{ appears infinitely often on } \tau\} = k$ is even and therefore $\tau$ is a $\mu$-trail. Yet this contradicts the fact that $\calG(\Phi)@\Gamma$ is winning for Builder.
\end{proof}

\setcounter{theorem}{\getrefnumber{thm.CompleteenessNWtwo}}
\addtocounter{theorem}{-1}
\begin{theorem}[Completeness]
	If a pure sequent $\Gamma$ is unsatisfiable, then $\Gamma$ is provable in \NWtwo.
\end{theorem}
\begin{proof}
	Follows by contraposition from Proposition \ref{prop.Completeness}.
\end{proof}

\subsection{Tracking automaton}\label{app.trackingAut}

\setcounter{theorem}{\getrefnumber{lem.slimTrails}}
\addtocounter{theorem}{-1}
\begin{lemma}
	Let $\pi$ be a saturated $\NWtwo$ proof of $\Gamma$. On every infinite branch of $\pi$ there is a slim $\mu$-trail.
\end{lemma}
\begin{proof}
	Let $\pi$ be a saturated \NWtwo proof of $\Gamma$. Let $\alpha$ be a branch of $\pi$ and $\tau$ be a $\mu$-trail on $\alpha$. 			
	For condition (ii) assume that there is an upward trail $(\eta x. \phi,\phi\subst{\eta x. \phi}{x},k) \in \sfT_{u,v}$ on $\tau$. The sequent $\sfS_v$ contains $\eta x. \phi \trace[k] \phi\subst{\eta x. \phi}{x}$ and $\eta x. \phi$, because the application of \RuEta is cumulative. Thus we can replace the trail $(\eta x. \phi,\phi\subst{\eta x. \phi}{x},k) \in \sfT_{u,v}$ in $\tau$ by $(\eta x. \phi,\eta x. \phi,1)(\eta x. \phi,\phi\subst{\eta x. \phi}{x},k) \in \sfT_{u,v}\sfT_{v,v}$.
	
	Regarding condition (i) we first assume that in $\tau$ there are only detour trails at nodes, that are labelled by a different rule than \RuTrans. This is not a restriction as all upward trail relations for the rule \RuTrans are of the form $(\phi,\phi,1)$, thus we can apply the same detour trail at its child. Assume that there is a subword of $\tau$ consisting of two detour trails $(\phi,\psi,k)(\psi,\chi,l)$, where $(\phi,\psi,k),(\psi,\chi,l) \in \sfT_{u,u}$. As $u$ is not labelled by \RuTrans and \RuTrans is always applied if applicable, also $(\phi,\chi,\max\{k,l\}) \in \sfT_{u,u}$. Hence we can replace $(\phi,\psi,k),(\psi,\chi,l)$ by $(\phi,\chi,\max\{k,l\})$. Doing this for all upward trails of the form $(\eta x. \phi, \phi\subst{\eta x. \phi}{x},k)$ and subwords consisting of two consecutive detour trails results in a slim $\mu$-trail $\tau'$.
\end{proof}

\subsection{Determinisation of $\epsilon$-parity automata}\label{app.determinisation}

\setcounter{theorem}{\getrefnumber{thm.determinisationCorrectness}}
\addtocounter{theorem}{-1}
\begin{theorem}
	The automaton $\bbA^S$ is equivalent to $\bbA$.
\end{theorem}
\begin{proof}
	``$\supseteq$'': Let $r = (a_0,n_0)(a_1,n_1)\cdots$ be an accepting run of $\bbA$ on some word $w$. We want to show that the unique run $\rho = S_0S_1\cdots$ of $\bbA^S$ on $w$ is accepting.
	
	We define a sequence of natural numbers $m(0)<m(1)<\cdots$ such that $m(j) = \max\{i \| n_i = j\}$ for $j \geq 0$. Intuitively $m(j)$ is the last index in the run $r$ such that $j$-many basic transitions were applied. In other words in the run $r$ at index $m(j)$ the $j+1$-th basic transition is applied.
	
	\claim{1}
	For every $j \in \omega$ there is a a unique stack $\tau_j$ such that 
	$a_{m(j)}^{\tau_j}$ is in the Safra-state $S_j$. 
	
	\claimproof{1}
	By induction on $j$. It holds that $a_0 = a_I$ and $m(j) = 0$, as we assume that $\Delta_{\epsilon}(a_I) = \emptyset$. By definition $S_0 = \{a_I^\epsilon\}$.
	
	Now assume that $a_{m(j)}^{\tau_j} \in S_j$. After step 1 of the transition function $a_{m(j)+1}^{\tau_j} \in S_j'$. In the run $r$ between $(a_{m(j)+1},j)$ and $(a_{m(j+1)},j)$ all transitions are $\epsilon$-transitions. Therefore after step 3 of the transition function $a_{m(j+1)}^{\tau'} \in S_j'$ for some $\tau'$. After that elements are removed such that we end up with an unique $\tau_{j+1}$ with $a_{m(j+1)}^{\tau_{j+1}} \in S_{j+1}$.
	\claimproofend
	
	We will now analyse the sequence $(\tau_j)_{j\in \omega}$. 
	Let $h \isdef \lim\inf |\tau_j|$, that is, $h$ is the maximal number such that 
	cofinitely many $\tau_{j}$ have size at least $h$.
	Let $J_0$ be such that $|\tau_j| \geq h$ for all $j \geq J_0$. 
	For $0 \leq l \leq h$ we let $\tau \fina{l}$ denote the stack consisting of the first $l$ names in $\tau$. We say that $\tau_j\fina{l}$ is constant for $j \geq J$ if for all $i,j \geq J$ it holds that $\tau_i\fina{l} = \tau_j\fina{l}$.
	
	\claim{2}
	There exists $J \in \omega$ such that $\tau_j\fina{h}$ is constant for $j \geq J$.
	
	\claimproof{2}
	By induction on $l$ we prove that there exist $J_l \geq J_0$ such that $\tau_j\fina{l}$ is constant for $j \geq J_l$, for all $0 \leq l \leq h$. 
	For $l = 0$ this is trivial. Now assume that it holds for $l < h$. 
	For simpler notation write $g \isdef J_l$, and let $\dx$ and $\sigma_{g}$ be such 
	that $\tau_g = \tau_g\fina{l} \cdot \dx \cdot \sigma_g$. 
	Let $\theta_j$ denote the control in the Safra-state $S_j$. 
	The only way that $\tau_j\fina{l+1}$ might change for $j \geq J_l$ is in step 4
	of the transition function, if $\tau_j = \tau_j\fina{l}\cdot \dy \cdot \sigma_j$ 
	with $\dy <_{\theta_j} \dx$. 
	As every newly introduced name is added as the last element in $\theta$ this 
	implies that already $\dy <_{\theta_g} \dx$. 
	If $\tau_i[l+1]$ changes again, then there is $\dz <_{\theta_i} \dy$, which already
	implies $\dz <_{\theta_g} \dy$ and so on. 
	As there are only finitely many names below $\dx$ in $<_{\theta_g}$ the stack 
	$\tau_j[l+1]$ can only change finitely often for $j \geq J_l$ and thus for
	some $J_{l+1} \geq J_l$ it must hold that $\tau_j\fina{l+1}$ is constant for 
	$j \geq J_{l+1}$.
	\claimproofend
	
	Let $J \in \omega$ be  as given in Claim 2 and let $\dx$ be the $h$-th name in $\tau_J$. For $j \geq J$ the name $\dx$ is always active. 
	The run $r$ is accepting, thus there is an even $k$ such that $\Omega(a_j) = k$ for infinitely many $j$ and $\Omega(a_j)\leq k$ for all $j \geq T$ from some time $T$ onwards. We assume that $J$ is picked big enough such that $J \geq m(T)$. Therefore for some $j \geq J$ a $k$-name $\dy$ is added to the stack $\tau_j$. But we have $|\tau_i| = h$ for some $i \geq j$, and this can only happen in step 4 of the transition function if $\dx$ was invisible and thus $\dx$ is coloured green. Note that this also implies that $\dx$ is a $k$-name. Repeating this argument yields that $\dx$ is active cofinitely often and is coloured green infinitely often in $\rho$.

	\bigskip
	``$\subseteq$'': Assume that there is an accepting run $\rho = S_0S_1\cdots$ of 
	$\bbA^S$ on $w$. 
	Let $\dx$ be a $k$-name, that is active cofinitely often and coloured green 
	infinitely often. Let $t(0) < t(1) < \cdots $ be the minimal indices such that $\dx$ is in play 
	in $S_j$ for every $j \geq t(0)$ and such that $\dx$ is green in $S_{t(i)}$ for every $i \in \omega$. 
	
	For $j \in \omega$ let $S_{t(j)} = (A_j, f_j, \theta_j,c_j)$.
	For $p,q \in \omega$ let $w[p,q)$ denote the segment $z_p\cdots z_{q-1}$ of the 
	infinite word $w= z_0z_1\cdots$. 
	In particular $w = w[0,t_0)\cdot w[t_0,t_1)\cdots$.
	Our goal is to find certain $[t(j),t({j+1}))$-labelled paths\footnote{That is, paths $a_0\cdots a_k$ in $\bbA$ on input $[t(j),t({j+1}))$ starting at state $a_0$.}
	in $\bbA$ which can be composed to obtain a successful run. 
	These will be formalized in the following claims. Afterwards we can glue together those paths to obtain an infinite run of $\bbA$ on the word $w$. 
	
For $j \in \omega$ let $B_j$ be the set of states in the macrostate $A_{j}$
which have $\dx$ in their stack. 
Formally, $B_j \isdef \{b \in A_j \| \dx \text{ occurs in } f_j(b)\}$.
	
	\claim{3}
	For every $a \in B_0$ there is an $w[0,t_0)$-labelled path from $a_I$ to $a$.
	
	\claimproof{3}
	For all $i = 0,\ldots,t(0)$ let $C_i \subseteq A$ be the macrostate in $S_i$. 
	For all $b \in C_{i+1}$ there is $a \in C_i$ such that there exists $c \in A$
	with $\Delta_b(a,c_i)=c$ and $b \in \eClos(a)$. 
	This follows from the definition of step 1 and 3 of the transition function. 
	The other steps only manipulate stacks but do not change macrostates.
	The claim then follows by induction.	
	\claimproofend
	
\claim{4}
For all $j > 0$ and all $b \in B_{j+1}$ there is a state $a \in B_j$ and a 
$w[t_j,t_{j+1})$-labelled path $c_0\cdots c_h$ with $a = c_0$, $b = c_h$
and	$\max\{\Omega(c_j) \mid i = 1,\ldots,h\} = k$.

\claimproof{4}
As in the proof of Claim 3 we can show that there is $a \in A_j$ and a 
$w[t_j,t_{j+1})$-labelled path $c_0\cdots c_h$ with $a = c_0$ and $b = c_h$. 
Because $\dx$ is in play in $S_j$ for all $j \geq t_0$ the name $\dx$ can never be 
	introduced in the transition function. 
	Thus we may conclude that $\dx$ was already present in the stack $\tau_j$ of $a$ 
	in $S_{t(j)}$, meaning that $a \in B_j$. It remains to show that there is such a path where $\max\{\Omega(c_j) \| i = 1,...,h\} = k$. In $S_{t(j)}$ the name $\dx$ is visible in all stacks, where $\dx$ occurs. In $S_{t(j+1)}$ the name $\dx$ is coloured green, indicating that after step 4 of the transition function in $S_{t(j+1)-1}'$ the name $\dx$ is invisible. This can only happen if a $k$-name $\dy$ was added to the stack $\tau_j$ in $S_{t(j)+1}...S_{t(j+1)}$ in step 2 or 3 of the transition function. Then $\Omega(c_j) = k$ for some $j = 1,...,h$. As $\dx$ is always in play we also have $\Omega(c_j) \leq k$ for all $j = 1,...,h$ and thus $\max\{\Omega(c_j) \| i = 1,...,h\} = k$.
	\claimproofend
	
	We will now glue paths together to obtain an infinite path through $\bbA$. This can be achieved using König's Lemma. Let $G = (V,E)$ where 
	\begin{align*}
		V = &\{a_I\} \cup \{(a,j) \| a \in B_j \text{ and } j \in \omega\},\\
		E = &\{(a_I,(a,0))\} \cup \\ 
		&\{((a,j),(b,j+1))\| b \in B_{j+1} \text{ and $a \in B_{j}$ as in Claim 4}\}
	\end{align*}
	Clearly $G$ is a connected, finitely branching and infinite graph. Hence we can apply König's Lemma to obtain an $w$-labelled path $r' = a_0a_1\cdots$ in $\bbA$, where $\Omega(a_j)\leq k$ for cofinitely many $j \in \omega$ and $\Omega(a_j) = k$ for infinitely many $j \in \omega$. In particular we find $r' \in \mathrm{Acc}$. By adding natural numbers $n_0, n_1,...$ in a straightforward way we obtain the accepting run $r = (a_0,n_0)(a_1,n_1)...$ of $\bbA$ on $w$, whose projection is $r'$. 
\end{proof}

\end{document}